\documentclass[12pt, letter]{article}
\renewcommand{\baselinestretch}{1.5}
\usepackage{amsmath}
\usepackage{epsfig}
\usepackage{amssymb}
\usepackage{amscd}
\usepackage{palatino}
\usepackage{graphics}
\usepackage{enumerate}
\usepackage{hhline}
\usepackage{natbib}
\usepackage{float}
\usepackage{chngpage}
\usepackage{rotating}
\usepackage{multirow}
\usepackage[table]{xcolor}
\usepackage[separate]{patty}
\def\argmax{\mathop{\rm argmax}\limits}
\def\max{\mathop{\rm max}\limits}
\usepackage{authblk}

\begin{document}
\bibliographystyle{apsr}
\renewcommand{\baselinestretch}{1}

\title{Algorithmic Fairness with Feedback}
\author{John W. Patty\thanks{Professor of Political Science and Quantitative Theory \& Methods, Emory University.  Email: \textit{jwpatty@gmail.com}.} \;\;\;\;\;
Elizabeth Maggie Penn\thanks{Professor of Political Science and Quantitative Theory \& Methods, Emory University.  Email: \textit{elizabeth.m.penn@gmail.com}.}}

\maketitle 

\begin{abstract}
    The field of algorithmic fairness has rapidly emerged over the past 15 years  as algorithms have become ubiquitous in everyday lives.  Algorithmic fairness traditionally considers statistical notions of fairness  algorithms might satisfy in decisions based on noisy data. We first show that these are theoretically disconnected from welfare-based notions of fairness. We then discuss two individual welfare-based notions of fairness, \textit{envy freeness} and \textit{prejudice freeness}, and establish conditions under which they are equivalent to \textit{error rate balance} and \textit{predictive parity}, respectively. We discuss the implications of these findings in light of the recently discovered impossibility theorem in algorithmic fairness (\cite{KleinbergMullainathanRaghavan16}, \cite{Chouldechova17}).\\[3ex]

    \centering \textsc{Working Paper: Comments, Suggestions, and Criticisms Welcome!}
\end{abstract}

\renewcommand{\baselinestretch}{1.5}

\newpage
\section{Introduction}

In 2016, ProPublica published a report on the validity of the COMPAS algorithm, which generates ``recidivism risk scores'' for consideration in important law enforcement decisions,\footnote{COMPAS stands for ''\textbf{C}orrectional \textbf{O}ffender \textbf{M}anagement \textbf{P}rofiling for \textbf{A}lternative \textbf{S}anctions.''} including setting bail and even whether to grant a prisoner parole.\footnote{Jeff Larson, Surya Mattu, Lauren Kirchner, and Julia Angwin. May 23, 2016.  ``How We Analyzed the COMPAS Recidivism Algorithm,'' \textit{ProPublica}: \textit{https://www.propublica.org/article/how-we-analyzed-the-compas-recidivism-algorithm}.}  The analysis concluded that Black defendants were more likely to be judged as ``high risk'' when they were low risk than were white defendants.  The company that developed and sells the algorithm, Northpointe (now Equivant), countered the ProPublica analysis, stating that Black and white defendants were equally likely to truly be high risk, conditional on being assigned a high risk COMPAS score.\footnote{William Dieterich, Christina Mendoza, and Tim Brennan. July 8, 2016. ``COMPAS Risk Scales: Demonstrating Accuracy Equity and Predictive Parity Performance of the COMPAS Risk Scales in Broward County.'' \textit{http://go.volarisgroup.com/rs/430-MBX-989/images/ProPublica\_Commentary\_Final\_070616.pdf}.}

From the outset, this debate clarified that there are multiple reasonable notions of fairness when decisions are ``noisy.'' The resulting discussion continues in both policy and scholarly communities. \cite{KleinbergMullainathanRaghavan16} and \cite{Chouldechova17} quickly and independently showed that, except for very special cases, the competing notions of fairness forwarded in this debate --- \textit{error rate balance} (ProPublica's measure) and \textit{predictive parity} (Northpointe's measure) --- can not be simultaneously satisfied.\footnote{A new direction in this regard is discussed in \cite{ReichVijaykumar21}, who show that if the data is sufficiently high quality for both groups, then pre-processing prior to classification can simultaneously equalize error rates across groups and satisfy calibration within each group.}  This impossibility result has led to many advances in the emerging field of \textit{algorithmic fairness} (or, \textit{algorithmic bias}).  Unsurprisingly, given the increasing ubiquity of algorithms in our everyday lives, many scholars have explored and are exploring ways to best judge whether and how an algorithm (or, indeed, any decision-making procedure) may be biased against (or in favor of) any particular group within a larger population to which the algorithm is applied.

Two aspects of this question that have received relatively little attention to date are (1) what are the connections between ``statistical notions of fairness'' and \textit{individual welfare} and (2) how should one judge fairness in the presence of \textbf{algorithmic endogeneity}: situations in which the data to which the algorithm is applied is itself a function of the algorithm?\footnote{Algorithmic endogeneity is a generalization of the concept of  \textit{performativity} (\textit{e.g.}, \cite{PerdomoEtAl20}).  Algorithmic endogeneity is more focused on the (theoretical) causal structure of ``data drift'' caused by an algorithm, as opposed to how to empirically estimate the structure of this drift.}  Each of these questions is important for extending the study of algorithmic fairness within social science.  We now provide a road map to our argument in this article.

\subsection{Outline of Our Argument}

Our focus in this article is \textit{algorithmic fairness}, a recently developed topic at the intersection of computer science, statistics, health, and social science.  At the heart of algorithmic fairness is the notion of different \textbf{groups} of individuals: algorithmic fairness is centrally focused on how and why any procedure (or ``algorithm'') treats members of the two groups differently.  Scholars have developed a number of fairness properties, the most prominent of which we discuss below prior to moving into our analysis.  A key feature of these notions for our purposes is that they are very generally defined, and essentially context-free.  They tend to be focused on either the \textit{inputs to} or the statistical properties of the \textit{decisions produced by} the algorithm.  Much of the focus in recent work has been on these statistical, decision-based notions of algorithmic fairness.  These notions differ with respect to what the notion of fairness is \textit{conditioned upon}.  In other words, different notions of algorithmic fairness (\textit{e.g.}, \textit{predictive parity} and \textit{error rate balance}, each discussed below) focus on different posterior probabilities, each of which represents a reasonable notion of fairness.  

A fundamental conclusion from the algorithmic fairness literature to date is that these various notions of fairness are in many cases incompatible with each other (\textit{e.g.}, \cite{KleinbergMullainathanRaghavan16}, \cite{Chouldechova17}).  However, the context-free nature of these notions present two separate but related challenges to applying algorithmic fairness in social science:
\begin{itemize}
    \item \textbf{Challenge 1:} \textit{Who benefits from pursuit of algorithmic fairness notion $X$?}
    \item \textbf{Challenge 2:} \textit{How does algorithmic fairness notion $X$ relate to individual incentives?}
\end{itemize}

This article attempts to extend our understanding of algorithmic fairness in each of these directions.  Both challenges require that we include some notion of \textbf{individual welfare} (or \textit{utility}).  We pursue this using the concept of \textbf{envy freeness}, which originated in the study of fair division and cooperative game theory.  In the context of algorithmic fairness, envy freeness is satisfied when no individual who \textit{will be classified by an algorithm} wishes to be in a different \textbf{group} of individuals (Definition \ref{Def:EnvyFree}).  We also introduce a related concept, \textbf{prejudice freeness}, which is satisfied when no individual who \textit{has been classified by an algorithm} wishes to be in a different \textbf{group} of individuals (Definition \ref{Def:PrejudiceFree}).

\paragraph{Our Results.} Tackling the two challenges above requires that we consider the \textbf{choices} available to any given individual.  When the details of an algorithm affect individual incentives, the situation under consideration exhibits \textbf{algorithmic endogeneity}.  We demonstrate that, if the individuals' preferences are independent of their group membership (\textbf{group independence}, Definition \ref{Def:GroupIndependence}), then any envy free algorithm in a binary classification problem must satisfy a choice-based fairness notion, which we refer to as \textbf{equal opportunity}.  We then demonstrate that, when group independence is satisfied in a binary classification problem, satisfaction of envy freeness implies satisfaction of a statistical notion of algorithm fairness known as \textbf{error rate balance} (Theorem \ref{Th:EnvyFreeEqualOpportunityIndependence}) and that the conditions of this equivalence are tight (Proposition \ref{Pr:EnvyFreeTightness}).\footnote{Note that error rate balance is closely related to the notion of equal opportunity defined in \cite{HardtPriceSrebro16}).}  

We then mirror this analysis for prejudice free algorithms. We show that, if individual preferences are group independent, any algorithm satisfying \textbf{predictive parity} must be prejudice free and, 
the two notions are equivalent if and only if the classification problem is binary.  


\subsection{Related Literature \label{Sec:LiteratureReview}}

The topic of algorithmic endogeneity has emerged over the past decade.  For example, \textit{strategic classification} (\cite{HardtEtAl16}) and \textit{performative prediction} (\cite{PerdomoEtAl20}) each provide guidance into the conditions under which an algorithm can accurately classify data that itself might be responsive to the details of the algorithm itself.  This literature has grown quickly, but our focus is different than many of the scholars who have contributed to it so far.  In particular, much of this literature is more focused on practical and technical questions about how to build (and train) the algorithm in such settings.  

This is important work, to be sure, but it is also largely context-free.  This abstraction is useful for some purposes, but also puts some distance between the analytic primitives of the model and theoretical concepts such as individual, group, and collective \textit{welfare} (however, see the discussion of \cite{ZafarValarRodriguezGummadiWeller17}, below).  Our goal here is to provide more structure on these theoretical questions in social science settings.  

Along these lines, our work is closely related to that of \cite{JungEtAl20}.  Specifically, \cite{JungEtAl20} present a canonical binary classification model of crime prevention with algorithmic endogeneity.  They show that, when designing the algorithm to minimize the probability that any individual will choose to commit a crime, the optimal algorithm will typically satisfy error rate balance.  Jung \textit{et al}.'s analysis focuses only on the \textit{behavioral} impact of the algorithm.  In other words, in the terminology of \cite{PattyPenn23}, authors suppose that the algorithm designer is \textit{compliance motivated}.  Our approach in this article removes this assumption.\footnote{In this regard we are following the approach also taken by \cite{PattyPenn23,PattyPenn23Bureaucracy}.}

In a model of employment with algorithmic endogeneity, \cite{PattyPenn21} consider the impact of information about applicants' group memberships.  They show that removing this information has very ambiguous welfare effects.  Removing this information can make both job applicants \& the employer better off, make both worse off, or benefit one ``side of the market'' while hurting the other.  The range of welfare effects is due to algorithmic endogeneity -- in particular, Patty \& Penn point out that algorithmic endogeneity might lead to there being \textit{multiple equilibria}, and the ``optimal algorithm'' will typically depend on which equilibrium the applicants and employer are engaged in.  Unlike both \cite{JungEtAl20} and our approach in this article, Patty and Penn do not allow the employer to ``pre-commit'' to a hiring algorithm.  The upside of this is that, in equilibrium, the employer's decisions satisfy a version of \textit{calibration}.

\subsection{Welfare-Based Notions of Fairness}

We present several related welfare-based concepts of fairness in this article, with the goal of connecting them with statistical-based notions of fairness, which are at the heart of many ongoing discussions in the field of algorithmic fairness.

\paragraph{Envy free Algorithms.}  We begin with the notion of \textbf{envy freeness}, which captures the notion that no individual regrets his or her own group membership \textit{prior to being classified} by the algorithm in question.  The general concept of envy freeness itself is classic and borrowed from the theory of fair division.  Unsurprisingly, we are not the first to leverage the concept of envy freeness to in algorithmic fairness.
 For example, \cite{ZafarValarRodriguezGummadiWeller17} use the concept to define the notions of \textit{preferred treatment} and \textit{preferred impact}.  Subsequently, \cite{BalcanDickNoothigattuProcaccia19} also proposed this concept in the context of algorithmic fairness.\footnote{\cite{BlandinKash21} also consider welfare-based fairness.}  Balcan \textit{et al}.'s focus is more on the computational and efficiency aspects of the notion in practice.\footnote{Along these lines, see also \cite{DworkImmorlicaKalaiLeiserson18} and \cite{UstunLiuParkes19}, who consider when and how decoupling can mitigate algorithmic unfairness.}  Our approach in this article is more theoretical/conceptual: we are interested in outlining the links between notions of statistical parity and individual welfare.  

\paragraph{Prejudice Freeness.}  In addition to envy freeness, we introduce a related concept, \textbf{prejudice freeness}, which captures the notion that no individual regrets his or her own group membership \textit{after being classified} by the algorithm in question.  Thus, we argue that envy freeness is an \textit{ex ante} fairness,\footnote{As we discuss below, envy freeness is arguably \textit{interim} fairness, because it is conditioned on the individual's characteristics.} while prejudice freeness captures an \textit{ex post} notion of fairness.  We are unaware of any theoretical work to date that mirrors ours in this regard.\footnote{The link between predictive parity and inference has recently attracted some attention (\textit{e.g.}, \cite{ZengDobrianCheng22}).}

\paragraph{Group versus Individual Fairness.} The algorithmic fairness literature can be roughly divided into two families of philosophical approaches to fairness: \textit{group-based} or \textit{individual-based} notions of fairness. In these terms, \cite{ZafarValarRodriguezGummadiWeller17} analyze group-based fairness.  This is one distinction between their approach and our approach.  The basic distinction in this regard is that our notion of fairness is conditioned on each individual's group membership and \textit{individual type}, as opposed to group membership alone: for reasons of space, this distinction is discussed in more detail in Appendix \ref{Sec:GroupLevelEnvyFreeness}.  We now introduce classification algorithms.

\section{Binary Classification Algorithms}

As common in the literature, we consider \textbf{binary classification problems} (or, \textit{BCP}s).  The primitives of such a BCP are as follows:
\begin{enumerate}
\item A binary set of \textbf{behaviors}, $\mathcal{B}=\{0,1\}$, with distribution
\[
\pi \equiv \Pr[\beta_i=1] \in [0,1].
\]
We will refer to the value $\pi$ as the \textbf{prevalence} within the population.
\item A binary set of \textbf{signals}, $\mathcal{S}=\{0,1\}$, with distribution \footnote{Note that Equation \eqref{Eq:SignalingDistribution} implies that the signal, $s$, is equally informative about either choice of behavior, $\beta \in \{0,1\}$.  This could be relaxed, but it simplifies the presentation and discussion of our results.}
\begin{equation}
\label{Eq:SignalingDistribution}
\Pr[s_i = \beta_i \mid \beta_i] \equiv \phi \in [\onehalf, 1].
\end{equation}
We will refer to the value $\phi$ as the \textbf{accuracy} of the signal.
\end{enumerate}

\paragraph{Inferring Behaviors from Signals.} In this article, we consider the problem of inferring each agent $i$'s behavior, $\beta_i$, from the signal that $i$'s behavior generates, $s_i$.  The question, then, is how to assign a \textbf{decision} of either $d_i=0$ or $d_i=1$ to each individual $i$, based on $s_i$.  Accordingly, an \textbf{algorithm} in this setting is a function $\delta: \mathcal{S} \to [0,1]$, where for any individual $j$ with signal $s_j \in \mathcal{S}$,
\[
\delta_{s_j} \equiv \Pr\left[d_j=s_j\mid s_j\right], 
\]
represents the conditional probability that individual $j$ is awarded a decision equal to the signal $s_j$ generated by $j$'s behavior, $\beta_j$.  We denote the set of algorithms by $\mathcal{A} \equiv [0,1]^2$.

\paragraph{Algorithmic Accuracy.}  The accuracy of an algorithm $\delta$'s classification decisions is a common benchmark in the algorithmic fairness literature.  In a BCP with prevalence $\pi \in [0,1]$ and measurement noise $\phi \in [0,1]$, the accuracy of any given algorithm, $\delta \in \mathcal{A}$, is defined by the confusion matrix displayed in Table \ref{Tab:ConfusionMatrixBaseline}.
\begin{table}[hbtp]
    \centering
    \begin{tabular}{|c|c|c|}
    \cline{2-3}
    \multicolumn{1}{c|}{}   &  $d_j=1$  & $d_j=0$   \\
    \hline
                            &                                               & \cellcolor{gray!25}       \\
         $\beta_j=1$        & 
         \begin{tabular}{c}
         \textbf{True Positive}\\
         \hline
         $W_1(\delta\mid \pi,\phi) \equiv$ \\
         $\pi\cdot (\phi \cdot \delta_1 + (1-\phi) \cdot (1-\delta_0))$
         \end{tabular}
         & \cellcolor{gray!25}\begin{tabular}{c}
         \textit{False Negative}   \\
         \hline 
         $L_1(\delta\mid \pi,\phi) \equiv$ \\
         $\pi\cdot (\phi \cdot (1-\delta_1) + (1-\phi) \cdot \delta_0)$
         \end{tabular}\\
                            &                                               & \cellcolor{gray!25}       \\
    \hline
                            & \cellcolor{gray!25}                           &                           \\
         $\beta_j=0$        & \cellcolor{gray!25} \begin{tabular}{c}
         \textit{False Positive}\\
         \hline
         $L_0(\delta\mid \pi,\phi) \equiv$ \\
         $(1-\pi)\cdot (\phi \cdot (1-\delta_0) + (1-\phi) \cdot \delta_1)$
         \end{tabular} 
         & \begin{tabular}{c}\textbf{True Negative}    \\
         \hline 
         $W_0(\delta\mid \pi,\phi) \equiv$ \\
         $(1-\pi)\cdot (\phi \cdot \delta_0 + (1-\phi) \cdot (1-\delta_1))$       
         \end{tabular}\\
                            & \cellcolor{gray!25}                           &                           \\
    \hline
    \end{tabular}
    \caption{Algorithmic Accuracy \& Probabilities of Each Outcome}
    \label{Tab:ConfusionMatrixBaseline}
\end{table}\\
The first of the two quantities, defined in the top left cell of Table \ref{Tab:ConfusionMatrixBaseline}, 
\[
W_1(\delta\mid \pi,\phi)\equiv \pi\cdot (\phi \cdot \delta_1 + (1-\phi) \cdot (1-\delta_0)),
\]
is the proportion of individuals who are correctly assigned a decision of $\delta_j=1$.  This is referred to as $\delta$'s \textbf{true positive rate}.  The second quantity, defined in the bottom right cell of Table \ref{Tab:ConfusionMatrixBaseline},
\[
W_0(\delta\mid \pi,\phi)\equiv (1-\pi)\cdot (\phi \cdot \delta_0 + (1-\phi) \cdot (1-\delta_1)),
\]
is the proportion of individuals who are correctly assigned a decision of $\delta_j=0$.  This is referred to as $\delta$'s \textbf{true negative rate}.  Finally, the off-diagonal quantities represent $\delta$'s \textbf{false positive rate} ($L_0$) and \textbf{false negative rate} ($L_1$).



\section{Algorithmic Fairness} 

To consider the fairness of an algorithm $\delta$, we suppose that that there are two \textbf{groups}, $g \in \mathcal{G} = \{X,Y\}$, each of which is associated with a prevalence, $\pi_g \in [0,1]$ and measurement noise, $\phi_g\in [\onehalf,1]$.  In this two group setting, the notion of an algorithm becomes more general, represented by a function $\delta: \mathcal{S} \times \mathcal{G} \to [0,1]$, where for any individual $j$ with signal $s_j \in \mathcal{S}$ and group membership $g_j \in \mathcal{G}$,
\[
\delta_{s_j}(g_j) \equiv \Pr\left[d_j=s_j\mid s_j,g_j\right], 
\]
represents the conditional probability that individual $j$ is awarded a decision equal to the signal $s_j$ generated by $j$'s behavior, $\beta_j$.\footnote{The algorithm is allowed to condition on each individual's group membership.  This is in line with much of the algorithmic fairness literature, which generally assumes that the group membership of any given individual is either directly observable or is possible to infer through ancillary data (``proxies'').} With the basic framework of algorithmic fairness laid out, we now proceed to discuss one of the more prominent notions of algorithmic fairness, \textit{error rate balance}.

\subsection{Error Rate Balance}

The notion of \textbf{error rate balance} (or ERB) is focused on whether members of one group can more accurately predict the decision he or she will be awarded, conditional on his or her behavior, than can members of another group.  Formally, ERB is defined as follows.

\begin{definition}[Error Rate Balance]
For any BCP, an algorithm $\delta$ satisfies \textbf{error rate balance} (ERB) with respect to the groups $X$ and $Y$ if
    \begin{eqnarray}
        \Pr[d_j=0\mid \beta_j=0, \delta, X] & = & \Pr[d_j=0\mid \beta_j=0, \delta, Y], \text{ and}  \label{Eq:PositiveErrorBalance}\\
        \Pr[d_j=1\mid \beta_j=1, \delta, X] & = & \Pr[d_j=1\mid \beta_j=1, \delta, Y]. \label{Eq:NegativeErrorBalance}
    \end{eqnarray}
\end{definition}
In the case of COMPAS, the scores did not satisfy ERB, because Black individuals who were \textit{low risk} to recidivate were significantly more likely to be assigned a \textit{high} recidivism risk score than were low risk white individuals.  

\paragraph{Error Rate Balance as Fairness to Individuals.}  We think of error rate balance as capturing ``fairness to the individuals,'' especially in the presence of algorithmic endogeneity, because each individual $j$ faces a trade-off when choosing their behavior, $\beta_j$.   Using an algorithm $\delta$ that satisfies ERB with respect to the distributions $(\pi_X,\pi_Y,\phi_X,\phi_Y)$ guarantees \textit{that \textit{any} individual $j$'s individual behavioral incentives induced by $\delta$ are independent of $j$'s group membership.}  This point will emerge again in our analysis below.  

\subsection{Predictive Parity}

The statistical notion of \textbf{predictive parity} (or PP) requires that decisions awarded to one group be equally informative about these individual's behaviors as those awarded to another group.  Formally, let 
\[
h: \mathcal{D} \times \mathcal{G} \times \mathcal{A} \times \left[\onehalf,1\right] \times [0,1]^2 \to [0,1]
\]
denote the posterior probability that $\beta_i=1$, conditional on $d \in \mathcal{D}$ and $g \in \mathcal{G}$, given the algorithm, $\delta \in \mathcal{G}$, the groups' signal accuracies, $(\phi_X,\phi_Y) \in [\onehalf,1]^2$, and the groups' prevalences, $(\pi_X,\pi_Y) \in [0,1]^2$. 
 With this in hand, PP is defined as follows.\footnote{For any BCP, an algorithm $\delta$ satisfies predictive parity if and only if
    \begin{eqnarray*}
        \Pr[\beta_j=0 \mid d_j=0, \delta, g] & = & \Pr[\beta_j=0 \mid d_j=0, \delta, g'], \text{ and}\\
        \Pr[\beta_j=1 \mid d_j=1, \delta, g] & = & \Pr[\beta_j=1 \mid d_j=1, \delta, g'],
    \end{eqnarray*}
    for all groups $g,g' \in \mathcal{G}$. 
} 
\begin{definition}[Predictive Parity]
An algorithm $\delta$ satisfies \textbf{predictive parity} (PP) with respect to groups $X$ and $Y$ if, for all $d \in \mathcal{D}$, the following holds:
    \begin{eqnarray*}
        h(d,X\mid \delta, \phi_X, \pi_X) & = &
        h(d,Y\mid \delta, \phi_Y, \pi_Y) \in \Delta(\mathcal{B}).
    \end{eqnarray*}
\end{definition}
In the case of COMPAS, the scores satisfied PP, implying that any individual who received a low-risk score was equally likely to actually be a low recidivism risk, regardless of his or her race.  

\paragraph{Predictive Parity as Fairness.} Substantively, PP requires that ``similar individuals who are awarded similarly should be similarly deserving.''  In line with this, PP can be thought of as an \textit{ex post} statistical notion of fairness.   In substantive terms, satisfaction of PP implies that an individual's group membership provides no information about his or her behavior that is not conveyed by the decision he or she was assigned by the algorithm.  
With ERB and PP defined, we now turn to considering the links between these notions and individual welfare.

\subsection{Algorithm Fairness \& Welfare}

 We now explicitly account for \textbf{individual welfare} from different behaviors and/or decisions.\footnote{This mirrors points recently raised by \cite{ZafarValeraGomezRodriguezGummadi17} and others.}  In order to compare algorithms on a welfare basis, one needs to make assumptions about individuals' preferences over their choices and/or the decisions they receive.
 
 Along these lines, it is important to note that people sometimes do not want the algorithm to be perfectly accurate.  For example, as demonstrated in the COMPAS example, an individual (\textit{e.g.}, individuals who would be accurately categorized as ``high risk'' for recidivism) might strictly prefer to be inaccurately categorized as low risk: essentially a preference of \textit{false negative outcomes}.  Similarly, a college applicant might strictly prefer to be admitted than otherwise (an example of a strict preference for \textit{false positive outcomes}).  More generally, there are situations, such as moral hazard problems and military conflict, in which an individual might simply have preference that his or her action is simply innacurately inferred.  The following example captures all such situations and demonstrates the tenuous nature of any claim that welfare is increasing in an algorithm's accuracy.
 
\begin{example}[Accuracy Is Not Necessarily Individually Optimal]
\label{Ex:RelativeFairness}

Consider an individual $i$ with preferences over behavior and decisions, $(\beta_i,d_i)$, as defined in Table \ref{Tab:ErrorPreferences}, where $\alpha_i, \gamma_i, \lambda_i,$ and $t_i$ are real numbers.
\begin{table}[hbtp]
\centering
\begin{tabular}{|c|c|c|}
    \cline{2-3}
    \multicolumn{1}{c|}{}   &  $d_i=1$          & $d_i=0$                   \\
    \hline
         $\beta_i=1$        & $\alpha_i$        & $\alpha_i + \lambda_i$    \\
    \hline
         $\beta_i=0$        & $t_i+\gamma_i$    & $t_i$                     \\
    \hline
    \end{tabular}
    \caption{Individual Preferences Over Outcomes \label{Tab:ErrorPreferencesRaw}}
\end{table}
The \textit{ex ante} expected payoff for individual $i$ from an algorithm $\delta$, given prevalence $\pi$ and measurement noise $\phi$, is equal to 
\begin{eqnarray*}
EU(\delta\mid \pi,\phi) 
& = & \pi \cdot 
\left(\alpha_i \cdot \bigg(1 + \phi \cdot (\delta_0+\delta_1) - \delta_0 - \phi  \bigg) 
+ (\alpha_i + \lambda_i) \cdot \bigg(\phi + \delta_0 -\phi \cdot (\delta_0+\delta_1) \bigg) \right) + \\
& & 
(1-\pi) \cdot \left(
(t_i+\gamma_i) \cdot \bigg(\phi + \delta_1 -\phi \cdot (\delta_0+\delta_1) \bigg) + t_i \cdot \bigg(1+\phi \cdot (\delta_0+\delta_1) - \phi - \delta_1  \bigg)
\right),
\end{eqnarray*}
By subtracting $\pi \cdot \alpha_i + (1-\pi) \cdot t_i$ from $EU$, we can re-express the preferences in Table \ref{Tab:ErrorPreferencesRaw} as defined in Table \ref{Tab:ErrorPreferences}.
\begin{table}[hbtp]
\centering
\begin{tabular}{|c|c|c|}
    \cline{2-3}
    \multicolumn{1}{c|}{}   &  $d_i=1$                  & $d_i=0$               \\
    \hline
         $\beta_i=1$        & 0                         & $\lambda_i$  \\
    \hline
         $\beta_i=0$        & $\gamma_i$                 & 0         \\
    \hline
    \end{tabular}
    \caption{Normalized Preferences Over Errors \label{Tab:ErrorPreferences}}
\end{table}\\
In words, individual $i$ receives a net payoff of $\gamma_i$ from a \textbf{false positive result} (\textit{i.e.}, $d_i=1$ despite $\beta_i=1$) and $\lambda_i$ from a \textbf{false negative result} (\textit{i.e.}, $d_i=0$ despite $\beta_i=0$).  Thus, the \textbf{normalized} \textit{ex ante} expected payoff for individual $i$ from an algorithm $\delta$, given prevalence $\pi$ and measurement noise $\phi$, is equal to the following:
\begin{eqnarray}
EU(\delta\mid \pi,\phi) 
& = & \pi \cdot \lambda_i \cdot \underbrace{\bigg(\phi + \delta_0 -\phi \cdot (\delta_0+\delta_1) \bigg)}_{FNR_i(\delta,\phi)} + (1-\pi) \cdot \gamma_i \cdot \underbrace{\bigg(\phi + \delta_1 -\phi \cdot (\delta_0+\delta_1) \bigg)}_{FPR_i(\delta,\phi)}, \label{Eq:NormalizedExpectedPayoffs}
\end{eqnarray}
where the labels underneath the terms of \eqref{Eq:NormalizedExpectedPayoffs} indicate the \textbf{false positive rate} (FPR($\delta,\phi$)) and \textbf{false negative rate} (FNR($\delta,\phi$)) faced by individual $i$ under algorithm $\delta$.

By construction, individual $i$'s normalized \textit{ex ante} \textbf{expected payoff} from a perfectly accurate algorithm $\delta$ is equal to 0.  However, without specifying the values of $\lambda_i$ and $\gamma_i$, \textit{it is impossible to judge whether the individual $i$ is better or worse off with an imperfectly accurate algorithm than $i$ would be with a perfectly accurate algorithm}.  

\end{example}
Example \ref{Ex:RelativeFairness} leads to the following simple result, which clarifies how the accuracy of an algorithm is related to its \textit{welfare effects}.\footnote{Note that Proposition \ref{Pr:AccuracyWelfare}'s hypothesis that the individual experiences the same cardinal level of gain and/or loss from errors (\textit{i.e.}, $|\lambda_i|=|\gamma_i|\neq 0$) is merely for expositional simplicity. The proofs of Proposition \ref{Pr:AccuracyWelfare} and all numbered results are contained in Appendix \ref{Sec:Proofs}.}
\begin{proposition}
\label{Pr:AccuracyWelfare}
    Suppose that $|\lambda_i|=|\gamma_i|\neq 0$.  Then $i$'s expected payoff is 
    \begin{itemize}
        \item increasing in the accuracy of $\delta$ if $\gamma_i=\lambda_i<0$,
        \item decreasing in the accuracy of $\delta$ if $\gamma_i=\lambda_i>0$, and
        \item not measurable with respect to accuracy of $\delta$ if $\lambda_i = - \gamma_i \neq 0$.
    \end{itemize}
\end{proposition}
  Because many statistical notions of fairness are based on accuracy, Proposition \ref{Pr:AccuracyWelfare} illustrates a general theoretical disconnect between \textit{statistical notions} and \textit{welfare-based notions} of algorithmic fairness.  Accordingly, we now turn to consider envy free algorithms.

\section{Welfare-Based Fairness: Envy Free Algorithms\label{Sec:EnvyFree}}

One welfare-based fairness notion, originating in cooperative game theory and the study of \textit{fair division}, is the concept of an algorithm being \textbf{envy free}.  An envy free algorithm is one in which, after any individual $i$ observes the algorithm, $\delta$, his or her group membership, $g_i$, and his or her type, $t_i$, individual $i$ does not want to change his or her group membership.  We now turn to the primitives of our analysis.

\subsection{The Primitives} 

The set of \textbf{individuals} is denoted by $N=\{1,2,\ldots,n\}$, and each individual $i \in N$ is assigned to one of two (observable) \textit{groups}, $g \in \mathcal{G} = \{X,Y\}$.\footnote{We consider only two groups in $\mathcal{G}$, but the number of groups is irrelevant for our analysis.}  For simplicity, we assume that $\mathcal{B}=\{0,1\}$ unless stated otherwise.\footnote{We present our analysis for cases with an arbitrary number of behaviors in Appendix \ref{Sec:GeneralClassification}.}  Each individual $i$ is endowed with a privately observed \textbf{type}, denoted by $t_i \in \mathcal{T} = \mathbf{R}$.  Each individual $i$'s type, $t_i$, is independently distributed, conditional on $i$'s group, $g_i$, according to a continuously differentiable cumulative distribution function (CDF), denoted by $F_{g_i}$, which we assume possesses full support on $\mathcal{T}$.


\paragraph{Timing of Behaviors and Decisions.}  After observing $g_i,t_i$, each individual $i \in N$ chooses exactly one \textbf{behavior}, $\beta_i \in \mathcal{B}$, after which a \textbf{signal} (denoted by $s_i \in \mathcal{S}$) is generated according to a CDF, $\phi_{g_i}$.  Finally, based on $i$'s signal and group, ($s_i, g_i$), $i$ will be assigned a \textbf{decision}, $d_i \in \mathcal{D}$.  In order to consider the welfare impacts of the algorithm $\delta$ that maps $(s_i,g_i)$ into a decision $d_i$, as well as to understand how each individual $i$ chooses his or her behavior, $\beta_i$, we now turn to describing individual preferences in this setting.

\paragraph{Individual Preferences.}  Example \ref{Ex:RelativeFairness} indicates the importance of being explicit about individual preferences (or welfare).  Without imposing any structure at this point, we represent any individual $i\in N$'s \textbf{preferences} by a \textbf{payoff function} as follows:\footnote{Notice that we are assuming that all individuals in the same group have the same preferences.  However, individuals in different groups can have different preferences.  Given that we have imposed no structure on the set of groups, this is without loss of generality, because one could assign each individual $i \in N$ to his or her own group.}
\begin{equation}
\label{Eq:UtilityGeneral}
u : \mathcal{B} \times \mathcal{D} \times \mathcal{T} \times \mathcal{G} \to \mathbf{R}.
\end{equation}
We assume that $u$ is continuous with respect to all of its parameters (specifically, $t$).\footnote{This is a restriction only with respect to $\mathcal{T}$, because we have assumed that $\mathcal{B},\mathcal{D}$, and $\mathcal{G}$ are finite sets.}  Equation \eqref{Eq:UtilityGeneral} is useful from a bookkeeping perspective, because it clarifies the four factors --- $i$'s behavior ($\beta_i$), the decision received by $i$ ($d_i$), $i$'s type ($t_i$), and $i$'s group membership, $g_i$ --- that might affect any individual $i$'s welfare in our framework.  

\begin{remark}
    \textit{Note that we are assuming that each individual $i$'s preferences over $\beta_i$ and $d_i$ are independent of every other individual $j$'s choice, $\beta_j$.  This rules out strategic interactions (and conveniently implies that the equilibrium is generically unique).}
\end{remark}

\paragraph{Expected Payoffs.}

For any algorithm, $\delta$, type $t \in \mathcal{T}$, and group $g \in \mathcal{G}$, individual $i$'s expected payoff from choosing action $\beta \in \mathcal{B}$ is:
\[
EU (\beta\mid \delta, t, g) \equiv E_{d \mid \delta_s(g)} \big[ E_{s|\beta,g}\left[u(\beta,d,t,g)\right] \big]
\]
We denote individual $i$'s optimal choice by $b^*$, defined as follows:\footnote{Our analysis is not necessarily assuming that individual choice is ``optimal'' in any specific sense, due to the well-known fact that \textit{any} single choice of behavior is consistent with a non-empty set of payoff functions, and we have not yet restricted $u$ at all.  We discuss this point in more detail in Appendix \ref{Sec:RationalChoice}.}
\begin{equation}
    \label{Eq:BestResponseGeneral}
    b^*(t,g\mid \delta,\phi) \in B^*(t,g\mid \delta,\phi) \equiv \argmax_{b \in \mathcal{B}} \left[EU (\beta_i\mid \delta, t, g) \right]
\end{equation}
We assume that $b^*$ is uniquely defined except on a set, as the exact selection in such cases will irrelevant for the purposes of this article.   We also assume that, for any algorithm $\delta \in \mathcal{A}$, for each behavior $\beta\in \mathcal{B}$, and each group $g \in \mathcal{G}$, there exists an nonempty open set of $\mathcal{T}$ such that $\beta$ is an optimal choice for members of group $g$. 

\begin{assumption}
For any algorithm, $\delta \in \mathcal{A}$, and any group $g \in \mathcal{G}$, 
\begin{enumerate}
    \item \underline{Optimal Behavior Is Generically Unique}: $F_g\left(\left\{t \in \mathcal{T}: |B^*(t,g\mid \delta,\phi)|=1\right\}\right)=1$, and
    \item \underline{Each Behavior Can Be Optimal}: $\forall \beta \in \mathcal{B}$, $F_g\left(\left\{t\in \mathcal{T}: b^*(t,g\mid \delta,\phi)=\beta\right\}\right) > 0$.
\end{enumerate}.
\end{assumption}

\paragraph{Equilibrium Expected Payoffs.}  We denote any individual's \textit{equilibrium expected payoff} from an algorithm $\delta$, given type $t\in \mathcal{T}$ and group membership, $g\in \mathcal{G}$ by the following:
\[
V^*(\delta, \phi\mid t,g) \equiv EU(b^*(t,g\mid \delta,\phi),\delta,t,g).
\]
We are now in a position to define envy free algorithms.

\subsection{Envy Freeness}

As mentioned in the introduction, an algorithm $\delta$ is \textbf{envy free} if no individual $i$ would strictly prefer to change his or her group membership, $g_i$, prior to being classified by $\delta$, supposing that $i$ will choose $\beta_i=b^*(t,g\mid \delta,\phi)$.  The notion is defined formally as follows.

\begin{definition}
\label{Def:EnvyFree}
    An algorithm $\delta$ is \textbf{envy free} if, for all $t \in \mathcal{T}$ and all $g \in \mathcal{G}$, 
    \begin{equation}
    \label{Eq:EnvyFreeDefinition}
    V^*(\delta, \phi\mid t,g) \geq V^*(\delta, \phi\mid t_i,g') \;\;\; \text{ for all } g' \in \mathcal{G}.
    \end{equation}
\end{definition}
Before continuing, note that Equation \eqref{Eq:EnvyFreeDefinition}
presumes that each individual $i$ makes his or her comparison of groups \textit{after} realizing his or her individual type, $t_i$.  Satisfaction of our definition of envy freeness is unaffected by the groups' type distributions (\textit{e.g.}, $F_X$ and $F_Y$).  In other words, our definition is an \textit{interim} concept, as opposed to an \textit{ex ante} definition.\footnote{We discuss this issue in some more detail in Appendix \ref{Sec:MeasurementErrorDiscussion}.}  We now describe a scope condition to clarify and simplify our analysis.

\subsection{Group Independent Preferences} 

Equation \ref{Eq:UtilityGeneral} allows for preferences over behavior ($\beta \in \mathcal{B}$) and decisions ($d\in \mathcal{D}$) to depend on both an individual's type ($t \in \mathcal{T}$) and group ($g \in \mathcal{G}$).  While this generality is useful in many ways, our analysis will require some structure on preferences.  One such structural assumption, following the general presumption in the algorithmic fairness literature that two individuals who differ only with respect to group membership are similar in all other respects, is \textbf{group independence}.  Preferences are group independent if each individual $i \in N$'s payoffs from behavior-decision pairs \textit{do not} depend on the individual's group membership, $g_i$.  This is defined formally as follows.

\begin{definition}
\label{Def:GroupIndependence}
    Individual preferences, $u : \mathcal{B} \times \mathcal{D} \times \mathcal{T} \times \mathcal{G} \to \mathbf{R}$ satisfy \textbf{group independence} if, for all $\beta\in \mathcal{B}$, $d\in \mathcal{D}$, $t \in \mathcal{T}$, and all $g,g' \in \mathcal{G}$,
    \[
    u(\beta,d,t,g) = u(\beta,d,t,g').
    \]
\end{definition}
Definition \ref{Def:GroupIndependence} is strong, but it sets aside some philosophical issues, such as interpersonal comparisons of individual welfare.  Furthermore, the assumption is in the spirit of much of the existing work in algorithmic fairness.  For example, \cite{DworkHardtPitassiReingoldZemel12} offer a compelling argument in favor of the idea that individuals ``in different groups'' are ``similar'' to one another from a fairness standpoint.\footnote{Furthermore, in the literature on envy free notions of fairness, the assumption is very similar to the ``$L$-Lipschitz condition'' in \cite{BalcanDickNoothigattuProcaccia19} (p.5).}  We adopt a similar approach in terms of welfare-based notions of fairness, adopting the idea that a fair algorithm should induce similar individuals to choose similar actions.

\subsection{Equal Opportunity \label{Sec:EqualOpportunity}}

We refer to our second welfare-based notion of fairness as \textbf{equal opportunity}, which captures the notion that \textit{an algorithm should induce similar individuals in the same situation to make identical behavioral choices}.  As mentioned above, we adopt the standpoint that two individuals who differ only in their group membership are presumptively similar to one another.  Accordingly, equal opportunity is formally defined as follows.

\begin{definition}[Equal Opportunity]
\label{Def:EqualOpportunity}
    An algorithm $\delta \in \mathcal{A}$ satisfies \textbf{equal opportunity} (EO) if 
    \[
    B^*(t,g\mid \delta,\phi) = B^*(t,g'\mid \delta,\phi).
    \]
\end{definition}
With equal opportunity defined, we now explore the connections between it, envy freeness, and error rate balance.

\subsection{Envy Freeness, Equal Opportunity, \& Error Rate Balance in BCPs}

Our first and most general result establishes that if individual preferences are group independent, then envy freeness, equal opportunity, and error rate balance are equivalent in \textit{any} binary classification problem.
\begin{theorem}
\label{Th:EnvyFreeEqualOpportunityIndependence}
    In any BCP, if individuals' preferences satisfy group independence, then
    \begin{enumerate}
        \item $\delta$ satisfies envy freeness only if $\delta$ satisfies EO and
        \item $\delta$ satisfies ERB if and only if $\delta$ satisfies envy freeness.
    \end{enumerate}
\end{theorem}

\paragraph{Discussion of Theorem \ref{Th:EnvyFreeEqualOpportunityIndependence}.} It is clear that group independence of preferences is central to the validity of Theorem \ref{Th:EnvyFreeEqualOpportunityIndependence}.   This dependence is arguably not too troubling from the perspective of algorithmic fairness, because group independence seems a natural starting point for motivating a search for ``fair'' algorithms. More surprisingly, Theorem \ref{Th:EnvyFreeEqualOpportunityIndependence} relies on the presumption that the classification problem is a \textit{binary}.  In classification problems with more than two decisions ($|\mathcal{D}|\geq 3$) there are always algorithms that satisfy equal opportunity but violate error rate balance. 

\begin{proposition}
\label{Pr:EnvyFreeTightness}
If the classification problem is ternary, or any higher order, then there exists an algorithm $\delta$ that satisfies equal opportunity, but violates ERB.
\end{proposition}

\paragraph{Discussion of Proposition \ref{Pr:EnvyFreeTightness}.}  Proposition \ref{Pr:EnvyFreeTightness} establishes that equivalence of envy freeness and error rate balance does not hold in higher order classification problems.  This is an important qualification of Theorem \ref{Th:EnvyFreeEqualOpportunityIndependence} because many contributions to algorithmic fairness focus on binary classification problems.  The reliance of Theorem \ref{Th:EnvyFreeEqualOpportunityIndependence} on the binary nature of the classification problem is in part due to the fact that a $k$-order classification problem has $k(k-1)$ ``types of errors,'' implying that the proper definition of ERB is unclear.\footnote{For asymmetric classification problems (where $|\mathcal{B}|\neq|\mathcal{B}|$, the number of errors in a classification problem is $|\mathcal{B}|\cdot (|\mathcal{D}|-1)$. We address general symmetric classification problems in Appendix \ref{Sec:GeneralClassification}.}  

\begin{remark}
    \textit{Note that any score that has more than two values falls within the bounds of Theorem \ref{Th:EnvyFreeEqualOpportunityIndependence} if and only if the set of scores can be divided into two mutually exhaustive subsets of scores that each contain yield equal lotteries over the set of decisions, $\mathcal{D}$.}
\end{remark}


\section{Prejudice Freeness}

Envy freeness and equal opportunity are each based on the idea that a fair algorithm should not differentially reward people who differ only with respect to their group membership, but both notions are focused on the \textit{upstream consequences} of an algorithm in the sense of how it affects individual incentives with respect to behavioral choices.  Many algorithms in the real-world have \textit{downstream consequences} as well.  These are consequences that are conditional on the decision rendered by the algorithm, holding the individual's behavioral choice, $\beta_i$, fixed.   

To consider the downstream consequences of any algorithm in a succinct fashion, suppose that each individual $i$ has an induced \textit{ex post} preference over $d \in D$ as follows:
\begin{equation}
\label{Eq:ExPostPayoffs}
v_i(d_i\mid \delta, g_i) = \sum_{\beta \in \mathcal{B}} \Pr[\beta\mid d_i, \delta, g_i] \cdot w_{g_i}(d_i,\beta),
\end{equation}
where $w_{g_i}(d_i,\beta)\in \mathbf{R}$ is exogenous and known for every $g_i\in \mathcal{G}, d_i\in \mathcal{D}$, and $\beta\in \mathcal{B}$.  In a binary classification setting, we can normalize and represent the \textit{ex post} payoff defined in Equation \eqref{Eq:ExPostPayoffs} more simply as follows:
\[
v_i(d_i\mid \delta, g_i) = \Pr[\beta_i=1\mid d_i, \delta, g_i] \cdot w_{g_i},
\]
where $w_{g_i}>0$ is exogenous and known for each $g \in \mathcal{G}$.  With these conditional \textit{ex post} payoffs in hand, we now define the notion of a \textbf{prejudice free} algorithm.

\begin{definition}
\label{Def:PrejudiceFree}
An algorithm $\delta$ is \textbf{prejudice free} (PF) if, for all $i \in N$, 
\[
v_i(d\mid \delta, g_i) \geq v_i(d\mid \delta, g') \;\;\; \text{ for all } g' \in \mathcal{G} \text{ and } d \in \mathcal{D}.
\]
\end{definition}
Just as an envy free algorithm $\delta$ makes individuals not want to switch their group membership, $g_i$, conditional upon individual $i$'s type, $t_i$, \textit{a prejudice free algorithm $\delta$ induces individuals to not gain from switching their group membership conditional on the decision, $d_i$, they are assigned by the algorithm.}

\paragraph{Equal Consequences.}  Letting $w\equiv \{w_g\}_{g\in \mathcal{G}}$ denote the profile of \textit{ex post} weights, we now define a condition on \textit{ex post} preferences that is analogous to group independence, above.
\begin{definition} 
\label{Def:EqualConsequences}
The profile of individual preferences, $w$, satisfies \textbf{equal consequences} if $w_g=w_{g'}$ for all $g,g' \in \mathcal{G}$.
\end{definition}
For reasons of space, equal consequences is effective ``black boxing'' the post-classification choices and incentives of individuals.  We present a more specified choice model providing a microfoundation for equal consequences in Appendix \ref{Sec:MicrofoundationPrejudiceFreeness}.  Leaving that for the interested reader, we can now formally state that equal consequences is necessary and sufficient for prejudice freeness and predictive parity to be equivalent in any binary classification problem.

\begin{theorem}
\label{Th:PrejudiceFreeEqualsPredictiveParity} 
Prejudice freeness and predictive parity are equivalent in any BCP if and only if $w=(w_X,w_Y)$ satisfies equal consequences.
\end{theorem}

The sufficiency of equal consequences is simple to see, but its necessity is a little more surprising.   That said, the necessity follows from the fact that equal consequences \& prejudice freeness are each based satisfying equalities.  The equivalence requires the problem to be a BCP because specifying one of the posterior beliefs (\textit{i.e.}, choosing a $d \in \mathcal{D}$ with $|\mathcal{D}|>2$) for a particular group does not identify the posterior beliefs for any of the other $|\mathcal{D}|-1$ decisions that the individual might be awarded.  We are now in a position to define equal consequences, which is analogous to the notion of equal opportunity discussed above.  We now turn to some conclusions and possible directions for future work.

\section{Conclusion}

Algorithms are increasingly used in the real-world to make both inferences about, and decisions regarding, individuals. 
 When this is the case, expectations about the algorithm --- and how it responds to changes in one's choices --- shape behavior and, accordingly, the data that the algorithm itself is charged with analyzing. This phenomenon of \textbf{algorithmic endogeneity} has begun to attract scholarly attention, but the connection between the starting points for concerns about algorithmic fairness and modern scientific theories of choice has been largely unexplored to date.  

 Addressing issues of algorithmic endogeneity requires that we focus on the incentives of individuals subject to, and individuals relying on decisions from, algorithms when making their own individual choices.  In this article, we have considered some of these issues, particularly by explicitly modeling the \textit{ex ante}/\textit{interim} decisions that individuals might need to make prior to an algorithm classifying individuals and/or the \textit{ex post} decisions one might make based upon the algorithm's classifications.  We have described two welfare-based notions of fairness, \textbf{envy freeness} and \textbf{prejudice freeness}. The former is a concept aimed at addressing \textit{fairness of the algorithm's upstream incentives}.  The latter, on the other hand,  considers \textit{fairness of the algorithm's downstream effects}.

  \paragraph{Envy Free Algorithms.} Envy freeness is a classic concept and we are not the first to propose it in the context of algorithmic fairness (\textit{e.g.}, \cite{ZafarValarRodriguezGummadiWeller17}, \cite{BalcanDickNoothigattuProcaccia19}).  We show that, in any BCP in which individuals' preferences do not depend on their group membership, envy freeness is equivalent to equal opportunity, which is equivalent to satisfaction of error rate balance (Theorem \ref{Th:EnvyFreeEqualOpportunityIndependence}).  

  In our framework, the only  \textit{observable} individual characteristic that we treat as ``manipulable'' (at least within the context of envy freeness) is one's group membership.  But the point we drawn from envy freeness is more general: the notion of envy freeness is essentially equivalent to what is known as \textit{strategy proofness} in the mechanism design literature.  Accordingly, envy freeness is an appropriate starting point for considering which (if any algorithms) eliminate individuals' incentives to ``lie to'' the algorithm. In this way, our results are connected to other parallel work (\textit{e.g.}, \cite{FrankelKartik19,FrankelKartik22}, \cite{DesseinFrankelKartik23}).  A general point in this regard is the link to our notion of error opportunity.  In binary problems, this fairness notion reduces to error rate balance.  
 
 \paragraph{Prejudice Free Algorithms.} Using a very similar logic to that underlying Theorem \ref{Th:EnvyFreeEqualOpportunityIndependence} we show that, if two groups have identical preferences over the inferences that are drawn about their behavior from the decision awarded by a binary classification algorithm, then the downstream consequences of the algorithm satisfy predictive parity if and only if the algorithm is prejudice free (Theorem \ref{Th:PrejudiceFreeEqualsPredictiveParity}).  While we conceive of Theorem \ref{Th:EnvyFreeEqualOpportunityIndependence} as speaking to debates about structural inequality (\cite{PennPatty23}), we think that Theorem \ref{Th:PrejudiceFreeEqualsPredictiveParity} is more closely related to the \textit{perpetuation} of inequality (see, for example, \cite{Fryer07}).  

 \paragraph{Another Impossibility Theorem.} Even with the caveats (group independence, group blindness, and binary classification problems) taken into account, it is nonetheless clear that, mirroring \cite{KleinbergMullainathanRaghavan16}, \cite{Chouldechova17}, our results imply that it is generally impossible for an algorothm to be both envy free and prejudice free.  We believe that our arguments both clarify the impossibility theorem's relevance for social scientists and demonstrate that the impossibility result is valid even in the face of algorithmic endogeneity.  Given the importance of the impossibility theorem in general (and particularly in terms of attracting our attention to the topic), it is useful to at least briefly describe a recent possibility result that illuminates the logic of the impossibility theorem.

 
\paragraph{Implications of \cite{ReichVijaykumar21}'s Possibility Result.}  Recent work by \cite{ReichVijaykumar21} suggests that Theorem \ref{Th:ImpossibilityImplication} might be misinterpreting \cite{KleinbergMullainathanRaghavan16} and \cite{Chouldechova17}'s main results but, honestly, we have not yet been able to gain a clear picture of the relationship between \cite{ReichVijaykumar21}'s possibility theorem and the impossibility result.  As we read \cite{ReichVijaykumar21}, the angle they take is to separate the classification process into two stages, where a calibrated score based on the observable traits of the individual in question (satisfying predictive parity) is handed off to a binary classifier (satisfying error rate balance) that uses only this score to assign a decision to the individual.  \cite{ReichVijaykumar21} demonstrate that this approach can, with sufficiently high quality data, at least partially sidestep the negative conclusions of the impossibility result.  What we are not entirely clear on at this point (and it's surely our own fault) is why such a ``two-step approach'' can not be represented as a ``one-step approach'' that fits within the results of either \cite{KleinbergMullainathanRaghavan16} or \cite{Chouldechova17}.

 \paragraph{Extending Preferences.} Our results impose a reasonable, but strong, condition on individuals' preferences.  Group independence and group blindness can each be violated in practice (\textit{e.g.}, medical care).  One could relax this condition by pursuing a more ``purely utilitarian'' approach (as in \cite{BlandinKash21}).  One issue with moving in this direction, however, is that measuring fairness only in  welfare terms while also allowing individual welfare to differ across groups imposes no limits on what one \textit{might} consider ``fair.''  

 \paragraph{Understanding Fights Over Fairness.} This leads to the final (and, in our opinion, more productive) avenue for future work: \textit{fairness} is a \textit{concept}.  Satisfaction of the concept is a theoretical matter, but \textit{raising} the issue of fairness (for example, to contest the failure to be awarded a pay raise, promotion, admission, \textit{etc.}) is simultaneously a theoretical and empirical matter. 
 This raises issue beyond ``transaction cost economics.''  While ``bringing up fairness'' in a negotiation is potentially costly, the bigger issue (to us) is \textit{where} the fairness arises from.  Our model (as in others, including \cite{BlandinKash21}) presumes that the quality of an outcome for an individual is independent of the behaviors chosen by \textit{other} individuals.  
 
 This assumption is not silly: many individual decisions are essentially decision-theoretic in the sense that the payoff from the decision in question is independent of other individuals' decisions.  But, there are many examples in which this independence is not satisfied, including hiring, traffic predictions, voter identification laws, and hurricane warnings.  For example, suppose that a city government learns that a hurricane will strike the city in 72 hours.  Should the city announce this?  Reasonable arguments suggest yet, but this leaves the secondary question of whether \textit{all} individuals should be warned at the same time.  If the city announces the alert to all at once, then it is reasonable to suppose that there will be a high likelihood of the alert causing traffic congestion on the road(s) out of town, possibly leading to fewer people being able to get out of time in time due to traffic jams and other congestion effects.

 Ultimately, the consideration of fairness in social science of decisions (algorithmic or otherwise) \textit{without} reference to the context of the decision itself, is quite possibly a hopeless endeavor.  We hope that our presentation and analysis above clarifies at least some of the aspects of the ``context'' that need to be taken into account to move the discussion of both algorithmic endogeneity and fairness beyond ``merely'' statistical measures of parity.

\newpage

\appendix

\noindent \textbf{Online Appendix for ``Algorithmic Fairness with Feedback''}\\
John W. Patty (Emory) and Elizabeth Maggie Penn (Emory)\\
\today

\section{Proofs of Numbered Results \label{Sec:Proofs}}

\noindent \textbf{Proposition \ref{Pr:AccuracyWelfare}}
    \textit{Suppose that $|\lambda_i|=|\gamma_i|\neq 0$.  Then $i$'s expected payoff is }
    \begin{itemize}
        \item \textit{increasing in the accuracy of $\delta$ if $\gamma_i=\lambda_i<0$,}
        \item \textit{decreasing in the accuracy of $\delta$ if $\gamma_i=\lambda_i>0$, and}
        \item \textit{not measurable with respect to accuracy alone if $\lambda_i = - \gamma_i \neq 0$.}
    \end{itemize}
\begin{proof}
    Follows from the argument in Example \ref{Ex:RelativeFairness} and Equation \eqref{Eq:NormalizedExpectedPayoffs}.
\end{proof}

\noindent \textbf{Theorem \ref{Th:EnvyFreeEqualOpportunityIndependence}}
    \textit{In any BCP, if individuals' preferences satisfy group independence, then
    \begin{enumerate}
        \item $\delta$ satisfies envy freeness only if $\delta$ satisfies EO and
        \item $\delta$ satisfies ERB if and only if $\delta$ satisfies envy freeness.
    \end{enumerate}}
\begin{proof}
    Follows from Theorem \ref{Th:EnvyFreeERB}, in Appendix \ref{Sec:GeneralClassification}.
\end{proof}

\noindent \textbf{Theorem \ref{Th:PrejudiceFreeEqualsPredictiveParity}}
\textit{Prejudice freeness and predictive parity are equivalent in any BCP if and only if $w=(w_X,w_Y)$ satisfies equal consequences.}
\begin{proof}
    
\end{proof}

\noindent \textbf{Proposition \ref{Pr:EnvyFreeTightness}}
    \textit{The hypotheses of Theorem \ref{Th:EnvyFreeEqualOpportunityIndependence} are tight.  If either or both of the following are relaxed:}
    \begin{enumerate}
        \item \textit{The binary classification structure of the decision problem and/or}
        \item \textit{The satisfaction of group independence by $u$,}
        \end{enumerate}
        \textit{then there exists an algorithm $\delta$ that satisfies equal opportunity, but violates ERB.}
\begin{proof}
The definition of ERB when the set of decisions is larger than two (\textit{i.e.}, $|\mathcal{D}|\geq 3$) is more demanding than for BCPs.  Supposing for simplicity that $\mathcal{B}=\mathcal{D}$ and that $\mathcal{G}=\{X,Y\}$, a general definition of ERB for an algorithm $\delta$ is as follows:
\[
\Pr[d_i=d\mid \beta_i,\delta,X] =
\Pr[d_i=d\mid \beta_i,\delta,Y], \;\;\; \text{ for all } (\beta_i,d) \in \mathcal{B} \times \mathcal{D}.
\]
Suppose that $\mathcal{B}=\mathcal{D}=\{1,2,\ldots,m\}$ with $m\geq 3$, and individual preferences, $u$, satisfy group independence as follows:
\[
u(\beta_i,d_i) = \begin{cases}
    1 & \text{ if } \beta_i=d_i=1,\\
    0 & \text{ otherwise.}
\end{cases}
\]
Then, suppose that the signal distribution (over $\mathcal{S}=\mathcal{B}$) is identical for the two groups and defined as follows for some $\phi\in \left(\frac{1}{m},1\right)$:
\[
\Pr[s_i = \beta \mid \beta_i] = \begin{cases}
\phi & \text{ if } \beta=\beta_i,\\
\frac{1-\phi}{m-1} & \text{ otherwise.}
\end{cases}
\]
Now, consider the algorithm 
\[
\delta_{s_i}(g) = \begin{cases}
  1 & \text{ if } d_i=s_i \text{ and } g=X,\\
  0 & \text{ if } d_i\neq s_i \text{ and } g=X,\\
  1 & \text{ if } d_i=s_i=1 \text{ and } g=Y,\\
  0 & \text{ if } d_i\neq s_i=1 \text{ and } g=Y,\\
  \frac{1}{m-1} & \text{ for each } d_i \neq s_i \in \mathcal{D} \; \text{ if } s_i \neq 1 \text{ and } g=Y.
\end{cases}
\]
The expected payoff under $\delta$ for an individual in group $X$ from choosing $\beta_i=1$ is equal to $\phi$.  His or her expected payoff under $\delta$ from choosing $\beta_i\neq 1$ is $\frac{1-\phi}{m-1}$.  On the other hand, the expected payoff under $\delta$ for an individual in group $Y$ from choosing $\beta_i=1$ is also equal to $\phi$.  His or her expected payoff under $\delta$ from choosing $\beta_i\neq 1$ is again also equal to $\frac{1-\phi}{m-1}$. 
Accordingly, $\delta$ satisfies equal opportunity (and envy freeness).  

However, $\delta$ does not satisfy ERB, because 
\begin{eqnarray*}
\Pr[d_i=\beta_i \mid \beta_i\neq 1, \delta,X] & = & \phi, \text{ but}\\
\Pr[d_i=\beta_i \mid \beta_i\neq 1, \delta,Y] & = & 0.
\end{eqnarray*}

Now suppose that $\mathcal{B} = \mathcal{D} = \mathcal{S} = \{0,1\}$, the signaling distribution satisfies 
\[
\Pr[s_i=\beta_i\mid \beta_i] = \phi \in (1/2,1),
\] 
and that individual preferences, $u$, \textit{violate} group independence as follows:
\[
u(\beta_i,d_i,g_i) = \begin{cases}
    1 & \text{ if } (\beta_i,d_i)=(0,0) \text{ and } G=X,\\
    -1 & \text{ if } (\beta_i,d_i)=(1,1) \text{ and } G=X,\\
    -1 & \text{ if } (\beta_i,d_i)=(0,0) \text{ and } G=Y,\\
    1 & \text{ if } (\beta_i,d_i)=(1,1) \text{ and } G=Y,\\
    0 & \text{ otherwise.} 
\end{cases}
\]
Now consider the following algorithm, $\hat{\delta}$:
\[
\hat{\delta}(g) = \begin{cases}
  (1,0) & \text{ if } g=X,\\
  (0,1) & \text{ if } g=Y.
\end{cases}
\]
This algorithm satisfies equal opportunity and envy freeness, but does not satisfy ERB.\footnote{As an aside, it is possible to show that equal opportunity is not satisfied by an envy free algorithm in this setting.}
\end{proof}

\section{Group Versus Individual Level Envy Freeness \label{Sec:GroupLevelEnvyFreeness}}

The distinction between group- and individual-level envy freeness can be captured by considering the following alternative, group-based, notion of envy freeness.  For each pair of groups $g_i, g \in \mathcal{G}$, let
\begin{equation}
    \label{Eq:GroupLevelUtility}
\overline{V}_g^*(\delta, \phi, g_i\mid u) \equiv \int_{\mathcal{T}} V^*(\delta, \phi\mid t,g) \; \dee F_{g_i}(t_i).
\end{equation}
denote the expected sequentially rational payoff for any individual $i$, conditional on his or her group membership, $g$. Then group-level envy freeness is defined in an analogous fashion to Definition \ref{Def:EnvyFree}) as follows:
\begin{definition}
\label{Def:EnvyFreeGroupLevel}
    An algorithm $\delta$ is \textbf{envy free at the group level} if, for each individual $i \in N$ with true group membership $g_i \in \mathcal{G}$,
    \[
    \overline{V}_{g_i}^*(\delta, \phi \mid u) \geq \overline{V}_{g'}^*(\delta, \phi \mid u) \text{ for all } g' \in \mathcal{G}.
    \]
\end{definition}

In addition to space constraints, one reason we do not pursue the implications of this definition in this article is that the definition of $\overline{V}_g^*$ in Equation \eqref{Eq:GroupLevelUtility} makes an important, and not necessarily sensible, choice.  Specifically, it implies that changing one's group \textit{does not change} the distribution of one's type, $t_i$.  This approach is tightly linked with similar questions about incentive compatibility and \textit{manipulability} in game theory and mechanism design, but also places a large burden on the researcher to clarify the timing of ``group switching.'' Perhaps an individual $i$ wants to join one group, $g$, rather than $i$'s own, $g_i$ \textit{because} $i$ is more likely to receive a ``good type'' if $i$ is a member of $g$ rather than $g_i$.  In substantive terms, the this is similar to comparing the question of whether someone would want to be ``born rich'' or ``born poor'' to the question of whether someone would rather ``grow up to be a doctor'' or ``grow up to be a lawyer'' --- our approach in the article is focusing on the second type of comparison, which is more clearly \textit{interim} in the sense that one's non-group specific type has been realized prior to one comparing the expected payoff of being a member of different groups.

\section{The Vacuity of ``Rational Choice'' \label{Sec:RationalChoice}}

Note that the relationship between $b^*$ and $t$ in Equation \eqref{Eq:BestResponseGeneral} is completelty unstructured.  This implies that, even though it is defined as an ``optimal choice'' for any individual, it has no necessary implications in terms of the ``rationality'' of any individuals.  In particular, Equation \eqref{Eq:BestResponseGeneral} is consistent with models in which each individual $i$ is ``hard-wired'' to choose an action completely based on his or her type, $t_i$. To see this, suppose that $\mathcal{B}=\{0,1\}$, $\mathcal{T}=\mathbf{R}$, and individual preferences are defined as follows:
\[
\tilde{u}_i(b,d,t,g) = b \cdot t.
\]
Given preferences $\tilde{u}_i$, individual $i$ should choose $\beta_i=1$ if $t_i>0$, and $\beta_i=0$ if $t_i<0$.  If $F_g$ is atomless, then this describes $i$'s ``optimal'' behavior with probability 1.

\section{Envy Freeness As An Interim Fairness Criterion \label{Sec:MeasurementErrorDiscussion}}

As discussed briefly in the body of the article, envy freeness is an interim concept of fairness because it presumes that each individual knows his or her type, $t_i$, when comparing his or her expected welfare from each possible group $g \in \mathcal{G}$.  In this appendix, we briefly discuss two conceptual issues with our definition.

\paragraph{\textit{Ex Ante} Envy Freeness.}  The following is an alternative, and more clearly \textit{ex ante}, definition of envy freeness.
\begin{definition}
\label{Def:ExAnteEnvyFree}
    An algorithm $\delta$ is \textbf{\textit{ex ante} envy free} if, for all $g \in \mathcal{G}$,
    \begin{equation}
    \label{Eq:ExAnteEnvyFreeDefinition}
    \int_{\mathcal{T}} V^*(\delta, \phi\mid t,g) \; \dee t \; \geq \int_{\mathcal{T}} V^*(\delta, \phi\mid t,g') \; \dee t \;\;\; \text{ for each } g' \in \mathcal{G}.
    \end{equation}
\end{definition}
Definition \ref{Def:ExAnteEnvyFree} is closer to the spirit of ideal theories of fairness than ours.  For example, it is arguably more similar to John Rawls's ``original position'' (\cite{Rawls71}) than the interim envy freeness defined in Definition \ref{Def:EnvyFree}.  However, Definition \ref{Def:ExAnteEnvyFree} is less useful for our purposes than the interim notion studied in the body of the article, because it is often impossible to satisfy when two or more groups face different type distributions (\textit{i.e.}, there are at least two groups $g, g' \in \mathcal{G}$ such that $F_g \neq F_{g'}$).  Because our theory treats these type distributions as exogenous and given, we leave the consideration of \textit{ex ante} envy freeness for future work.\footnote{Formally, any such exploration would need to impose some structure on how the designer (for example) might shape groups' type distributions.}  We now turn to a similar issue that is not entirely set aside by focusing on \textit{interim} envy freeness: how one should deal with groups generating signals, $s_i$, with heterogeneous accuracy, $\phi_g$.

\paragraph{Envy Freeness and Measurement Error.} Definition \ref{Def:EnvyFree} --- through the definition of $b^*$ in \eqref{Eq:BestResponseGeneral} --- presumes that if an individual's group membership changes from $g$ to $g'$, the the conditional distribution of the signal about $\beta_i$, $s_i$, does as well.  In other words, we are assuming that changing one's group from $g$ to $g'$ implies that $i$'s signal $s_i$ inherits the measurement noise of the new group, $g'$.

It would be interesting to consider an alternative notion of envy freeness in which individuals' measurement errors are heterogeneous across groups, but independent of their choice of which group to adopt. One issue with our definition of envy freeness can ultimately allow individuals to choose a group based on an \textit{ex post} comparison of sampling distributions, because one or more groups might have a \emph{degenerate} sampling distribution.

    However, without this assumption, envy freeness is arguably muddled from a conceptual standpoint.  To see this, consider any algorithm $\rho$ an individual $i$ who strictly prefers that the algorithm be accurate would benefit from changing his or her group membership to that of a group subject to a lower level of measurement noise so that the algorithm is ``mis-calibrated'' with respect to $i$'s true level of measurement noise.

Requiring that an individual compare group membership with their corresponding measurement noise reduces the desire to change groups to essentially ``cheat the algorithm.''' For example, our assumption about the measurement noise as ruling out the possibility of a lawyer ``becoming a medical doctor'' by taking the bar exam instead of the medical board exam.  

\section{The Disconnect between ERB and Envy Freeness}

The following proposition further clarifies the tenuousness of the link between welfare and error rate balance.
\begin{proposition}
\label{Pr:EnvyFreenessERBNotEquivalent}
    For any algorithm, $\delta$, satisfaction of envy freeness neither implies, nor is implied by, satisfaction of error rate balance by $\delta$. 
\end{proposition}
\begin{proof}
We first relax group independence of $u$. Consider the following setting, extending that considered in Example \ref{Ex:RelativeFairness}.

Suppose that there are two individuals, denoted by $X$ and $Y$ and suppose that $\phi_X=\phi_Y=1$.\footnote{This is for simplicity and sufficient for our purposes.}  Note that $\lambda_i$ and $\gamma_i$ are arbitrary numbers in Example \ref{Ex:RelativeFairness}, so that if 
\begin{eqnarray*}
    \lambda_X = 1 & \text{ and } & \gamma_X=0,\\
    \lambda_Y = 0 & \text{ and } & \gamma_Y=1.
\end{eqnarray*}
With this in hand, consider the following uniquely Pareto efficient algorithm:
\[
\hat{\delta}(g) = \begin{cases}
    (1,1) & \text{ if } g=X,\\
    (0,0) & \text{ if } g=Y.
\end{cases}
\]
Because $\hat{\delta}$ is insensitive to any individual $i$'s signal, $s_i$. No individual $i$'s expected payoff from $\hat{\delta}$ is a function of $i$'s group's level of measurement noise, $\phi_{g_i}$.  Furthermore, every individual $i$ should choose $\beta_i$ purely as a function of $i$'s type, $t_i$.  Accordingly, if $\pi_X \neq \pi_Y$, then ERB will be violated by $\hat{\delta}$, but $\hat{\delta}$ is envy free, because it guarantees each individual in each group his or her maximum possible \textit{ex post} payoff (conditional on $t_i$).

We now relax the supposition that $u$ is strict.  Suppose that there are two individuals, denoted by $X$ and $Y$ and suppose that  $\phi_X\leq \phi_Y=1$.  Furthermore, suppose that $\lambda_X=\lambda_Y=0$ and $\gamma_X=\gamma_Y=2$.
(These preferences do not satisfy strictness, because $\lambda_X=\lambda_Y=0$, but do satisfy group independence.)  

Then an envy free algorithm in this case is $\delta(X)=(0,1)$ and $\delta(Y)=(\onehalf,\onehalf)$, which yields every individual  in group $Y$ an expected payoff equal to 1.  However, this algorithm does not satisfy ERB, unless $\pi_X=\pi_Y$
\end{proof}

\section{A Microfoundation for Prejudice freeness \label{Sec:MicrofoundationPrejudiceFreeness}}

In the spirit of envy freeness, one can flip the problem and ask whether someone who is tasked with reviewing the algorithm's decisions would be interested in knowing the group membership of the individual who received a decision.  This can be described as whether the group membership of an individual would have a non-zero impact on the \textit{reliability} of the decision received by the individual.  

In our framework, the reliability of any decision $d$ produced by an algorithm $\delta$. is a function of the posterior beliefs regarding the behavior chosen by the individual who received decision $d$, given his or her group membership, $g$, measurement noise $\phi_g$, and the prevalence of different behavioral choices by membership of group $g$, $\pi_g$.  Along these lines, let
\[
h(d,g\mid \delta, \phi_g, \pi_g) \in \Delta(\mathcal{B})
\]
denote the conditional distribution of $\beta \in \mathcal{B}$, given decision $d \in \mathcal{D}$, group $g \in \mathcal{G}$, measurement noise $\phi_g \in \mathcal{B} \to \Delta(\mathcal{S})$, and group prevalence, $\pi_g \in \Delta(\mathcal{B})$.\footnote{A minor technical note is in order: we are assuming that the ex ante probability that any individual $i$ from any group $g$ might be assigned any given decision $d \in \mathcal{D}$ is strictly positive. Relaxing this is straight-forward but cumbersome and beyond our focus in this article.}  Furthermore, for any $\beta \in \mathcal{B}$, let 
\[
h_{\beta}(d_i,g_i\mid \delta,\phi_{g_i},\pi_{g_i}) \equiv \Pr[\beta_i=\beta\mid d_i, g_i, \delta, \phi_{g_i}, \pi_{g_i}]
\]
denote the posterior probability that individual $i$ chose $\beta_i=\beta$, conditional on $\delta$ assigning a decision of $d_i$ to $i$ and individual $i$'s group membership, $g_i$.  Note that, for any decision $d \in \mathcal{D}$ such that
\[
h(d,X\mid \delta, \phi_X, \pi_X) 
\neq 
h(d,Y\mid \delta, \phi_Y, \pi_Y), 
\]
then the reviewer's inferences about the behavior chosen by a member of group $X$ after observing decision $d$ awarded to that individual will differ from what the reviewer should infer if they observe the same decision, $d$, awarded to a member of group $Y$.  

\paragraph{Algorithms, Reviewers and Consequences.} The notion of a reviewer is left abstract in our presentation, but there are many real-world examples:
\begin{enumerate}
    \item a loan officer deciding whether to grant a loan to individual $i$ (\textit{e.g.}, $x_i = 1$) or deny it ($x_i=0$), based on the algorithm's recommendation, $d_i$, and individual $i$'s group membership, $g_i$,
    \item a judge deciding the bail level ($x_i \geq 0$) to require for pre-trail release of a defendant $i$, given $i$'s group membership $g_i$ and the risk score assigned to $i$, $d_i$,
    \item an employer deciding whether to hire an applicant from group $g_i$, given the hiring algorithm's recommendation, $d_i$.
\end{enumerate}
In all such cases, the reviewer is potentially interested in how much information about the individual's choice of behavior is contained in the decision reached by the algorithm.  With this in hand, we now define a third statistical notion of fairness, \textit{predictive parity}.

\subsection{The Reviewer's Problem}

We assume now that there is a  \textbf{reviewer}, $R$, who observes the decision each individual $i$ is awarded, $d_i$, by the algorithm $\delta$.  The reviewer knows $g_i$, $\phi_{g_i}$, and the prevalence of behavior within group $g_i$, denoted by $\pi_{g_i} \in \Delta(\mathcal{B})$.  With this information, $R$ then chooses a \textbf{consequence}, $x_i \in \mathcal{X}$, to assign to individual $i$.   For simplicity, we assume that $\mathcal{X}=\mathcal{D}$: the set of consequences from which the reviewer can choose is equal to the set of decisions that the algorithm can award to any given individual.\footnote{This also aligns with the algorithm $\delta$ ``recommending'' a particular consequence to the reviewer.} 

\subsection{The Reviewer's Payoff} 

We assume that the reviewer (denoted by $R$) has a payoff function, $u_R: \mathcal{X} \times \mathcal{B} \times \mathcal{D} \times \mathcal{G} \to \mathbf{R}$.  Thus, the payoff that the reviewer receives from awarding a consequence $x$ to an individual $i$ from group $g_i$ who chose $\beta_i$ and received a decision of $d_i$ from algorithm $\delta$ is denoted by $u(x,\beta_i,d_i,g_i)$.

\paragraph{The Reviewer's Optimal Choice.}. Based on the reviewer's payoff function, $u_R$, and the posterior beliefs induced by $\delta$ (given $\phi$ and $\pi$), $h$, the reviewer's \textit{optimal consequence} for any individual from group $g$ who was awarded decision $d$ is defined as follows.

\begin{equation}    \label{Eq:ReviewerBestResponseGeneral}
    x^*(d,g\mid \delta,\phi,\pi) \in X^*(d,g\mid \delta,\phi,\pi) \equiv \argmax_{x \in \mathcal{X}} E \bigg[ u_R(x,\beta,d,g) \; \bigg\mid \; \beta\sim h(d,g\mid \delta,\phi,\pi) \bigg]
\end{equation}
Note that we assume that $\mathcal{X}$ is finite, implying that $X^*(d,g,\mid \delta,\phi,\pi)$ is nonempty for all $(d,g) \in \mathcal{D}\times \mathcal{G}$.

\subsection{Prejudice Freeness}

With the reviewer's optimal choice, $x^*: \mathcal{D} \times \mathcal{G} \to \mathcal{X}$, defined, we now define 
\[
V_R^*(d, g\mid \delta, \phi, \pi) = E\bigg[ u_R(x^*(d,g),\beta,d,g) \; \bigg\mid \; \beta\sim h(d,g\mid \delta,\phi,\pi) \bigg]
\]
as the reviewer's \textbf{expected payoff} from choosing a consequence for an individual an algorithm $\delta$, given decision $d \in \mathcal{D}$, group membership $g$, measurement noise $\phi_g$, and prevalence $\pi_g\in \Delta(\mathcal{B})$.  We refer to any algorithm whose rewards are equally accurate for members of different groups --- \textit{which will imply that the expected payoff from choosing consequences for an individual to be independent of that group's membership} --- as \textbf{prejudice free}.  This notion is defined formally as follows.
\begin{definition}
    \label{Def:PrejudiceFreeAppendix}
    An algorithm $\delta$ is \textbf{prejudice-free} if, for all $d \in \mathcal{D}$, and all $g,g' \in \mathcal{G}$, 
    \[
    E_{d\mid \delta, \phi, \pi} \bigg[ V_R^*(d, g\mid \delta, \phi, \pi) \bigg] = E_{d\mid \delta, \phi, \pi} \bigg[V_R^*(d, g'\mid \delta, \phi, \pi)\bigg]
    \]
\end{definition}
When faced with a prejudice free algorithm, the reviewer $R$ is indifferent about the group membership of the individual whose decision he or she is given to review, given $\delta$, $\phi$, and $\pi$.  Definition \ref{Def:PrejudiceFree} (intentionally) mirrors envy freeness (Definition \ref{Def:EnvyFree}).  The only difference is that we do not incorporate any uncertainty about the reviewer's preferences (\textit{i.e.}, in our framework, the reviewer has no analogue for the individual's type, $t_i$).

\subsection{Group Blind Reviewer} 

With the notion of a prejudice free algorithm defined, it is useful to consider cases in which the reviewer's payoff, $u_R$, is independent of any individual $i$'s group membership, $g_i$.  We refer to this as \textbf{group blindness} and define it formally as follows.

\begin{definition}
    \label{Def:GroupBlind}
    The reviewer's payoff function, $u_R$, is \textbf{group blind} if, for all $x \in \mathcal{X}$, $\beta \in \mathcal{B}$, $d \in \mathcal{D}$, 
    \[
    u_R(x,\beta,d,g) = u_R(x,\beta,d,g') \text{ for all } g,g' \in \mathcal{G}.
    \]
\end{definition}

Group blindness for the reviewer mirrors group independence for the individuals in the sense that there is little reason to suspect that a reviewer whose payoffs are \textit{not} group blind should treat individuals from different groups equally.  Put another way, a group blind reviewer has no inherent desire for the algorithm to render decisions with different levels of precision for otherwise similar members of different groups.

\subsection{Equal Consequences}

An algorithm $\delta$ that satisfies \textbf{equal consequences} renders the group membership of any individual $i$ irrelevant to the reviewer's assignment of a consequence to $i$ after observing the decision $i$ was awarded by $\delta$.  This is defined formally as follows.  
\begin{definition}[Equal Consequences]
\label{Def:EqualConsequences}
    An algorithm $\delta \in \mathcal{A}$ satisfies \textbf{equal consequences} (with respect to $\phi=\{\phi_g\}_{g\in \mathcal{G}}$ and $\pi=\{\pi_g\}_{g\in\mathcal{G}}$) if there exists a selection $x^*: \mathcal{D} \times \mathcal{G} \to \mathcal{X}$ with $x^* \in X^*(d,g\mid \delta, \delta,\phi)$  such that, for all $g,g' \in \mathcal{G}$,
    \[
    x^*(d,g \mid \delta,\phi,\pi) = x^*(d,g' \mid \delta,\phi,\pi).
    \]
\end{definition}

\subsection{Prejudice Freeness \& Predictive Parity}

We now conclude our analysis by demonstrating a strong similarity between the reviewer's problem and the individuals' problems.  Specifically, if the reviewer is group blind, then the set of algorithms that satisfy prejudice freeness is equal to the set of algorithms that satisfy equal consequences.  Furthermore, in any BCP, the set of algorithms that satisfy prejudice freeness are equivalent to the set of algorithms that satisfy predictive parity.  These results are stated formally in the following.

\begin{proposition}
\label{Th:PrejudiceFreeEqualsPredictiveParity}
    If the receiver's preferences, $u_R$ are group blind, then the following are equivalent:
    \begin{enumerate}
        \item $\delta$ satisfies prejudice freeness, and 
        \item $\delta$ satisfies equal consequences.
    \end{enumerate}
    Furthermore, if each individual has two behaviors to choose between ($|\mathcal{B}|=2$), then the following are equivalent for any algorithm, $\delta$:
    \begin{enumerate}
        \item $\delta$ satisfies prejudice freeness, 
        \item $\delta$ satisfies equal consequences, and
        \item $\delta$ satisfies predictive parity.
    \end{enumerate}
\end{proposition}
\begin{proof}
    Fix any classification problem with consequences $\mathcal{X}$, suppose that the reviewer's preferences are group blind and, for simplicity, suppose that there are two groups: $\mathcal{G}=\{X,Y\}$.  

    Group blindness of $u_R$ and prejudice freeness of $\delta$ jointly imply that $\delta$ also satisfies equal consequences. To see this, note that group blindness of $u_R$ implies that there is a selection of optimal choices, $x^*$, such 
    \[
    \bigg[ V_R^*(d,g\mid \delta, \phi, \pi) = V_R^*(d,g\mid \delta, \phi, \pi) \bigg] \Rightarrow x^*(d,g) = x^*(d,g').
    \]
    The same argument can be used to establish that group blindness of $u_R$ and satisfaction of equal consequences by $\delta$ jointly imply prejudice freeness of $\delta$.

    Now we show that, if set of behaviors is binary (without loss of generality, let $\mathcal{B} = \{0,1\}$), then group blindness of $u_R$ and prejudice freeness of  $\delta$ imply that $\delta$ satisfies predictive parity. Note that, in a binary behavioral choice setting, one can normalize $u_R$ as follows for all decision $d\in \mathcal{D}$ and groups $g \in \mathcal{G}$:
    \begin{eqnarray*}
    \bar{u}_R(x,0,d,g) & \equiv & 0, \\
    \bar{u}_R(x,1,d,g) & \equiv & u_R(x,1,d,g) - u_R(x,0,d,g).
    \end{eqnarray*}
    The reviewer's net expected payoff from assigning any consequence $x_i \in \mathcal{X}$ to individual $i$, conditional on $(\tilde{d}_i,g_i)$, is equal to 
    \[
    \bar{u}(x,1,\tilde{d}_i,g_i) \cdot
    h_1(\tilde{d}_i,g_i\mid \delta, \phi_{g_i}, \pi_{g_i}).
    \]
    Furthermore, with the assumptions in hand, we can suppose that $\tilde{u}(x,1,d,g)\equiv 1$.
    
    Group blindness of $u_R$ implies that 
    \[
    \bar{u}(x,1,\tilde{d}_i,g) = \bar{u}(x,1,\tilde{d}_i,g') = \tilde{u}(x,1,\tilde{d}_i) \cdot
    h(\tilde{d}_i,g_i\mid \delta, \phi, \pi)
    \]
    for some real-valued function $\tilde{u}: \mathcal{X} \times \mathcal{B} \times \mathcal{D}$.  Accordingly, group blindness of $u_R$ implies that $x^*(d,g)$ is any function that maximizes $h(d,g\mid \delta, \phi_g, \pi_g)$. From this, it follows that any prejudice free algorithm $\delta$ must satisfy the following:
    \begin{eqnarray*}
    h_0(d,g\mid \delta, \phi_g, \pi_g) & = & 
    h_0(d,g'\mid \delta, \phi_{g'}, \pi_{g'}), \text{ and}\\
    h_1(d,g\mid \delta, \phi_g, \pi_g) & = & 
    h_1(d,g'\mid \delta, \phi_{g'}, \pi_{g'}),    
    \end{eqnarray*}
    implying that $\delta$ satisfies predictive parity.  Finally, to see that satisfaction of predictive parity implies that the algorithm $\delta$ is prejudice free whenever the reviewer is group blind, just reverse the steps above.

    Thus, if the reviewer is group blind, prejudice freeness and equal consequences are equivalent and, futhermore, in any BCP, these are equivalent to satisfaction of predictive parity, as was to be shown.
\end{proof}
Note that the final conclusion of Theorem \ref{Th:PrejudiceFreeEqualsPredictiveParity} requires that we are considering a BCP. This is for the same reason it is required in Theorem \ref{Th:EnvyFreeEqualOpportunityIndependence}: a single individual choice only partially identifies ordinal individual preferences in situations with more than two choice options.  

One could generalize the notions of PP and/or ERB to partially avoid this restriction, but (to our knowledge) these generalized notions are by no means standard in statistical decision theory.  This is even more true once one recalls the noisiness built into these settings (for example, our assumptions about $u$ imply that any given individual $i$ has a positive \textit{ex ante} probability of choosing any given behavior $\beta_i$ under \textit{any} algorithm $\delta$.) This noisiness implies that the space of probabilities that any given individual (when faced with more than 2 possible decisions) or any given receiver (when faced with individuals who chose from more than 2 possible behaviors) is confronted with a multidimensional space of potential probability distributions against which to optimize.  Similarly, substantive restrictions on the individuals' and/or reviewer's preferences might yield more results in ternary or higher-order classification problems, but at least at first blush, our ideas in this direction are essentially equivalent to assuming that individuals behave as if they are in a BCP from an \textit{ex post} welfare perspective.

\section{The Binary Choice Model \label{Sec:ModelAppendix}}

\paragraph{Individuals.}  The set of individuals is denoted by $N$, and we assume that $N$ is divided into two \textbf{groups}, which we denote by $X$ and $Y$.  For analytical simplicity, we assume that $X$ and $Y$ are of equal size and continuous: $X=Y=[0,1]$.  The \textbf{group membership} of any individual $i \in N \equiv X \cup Y$ is denoted by $g_i \in \mathcal{G} \equiv \{X,Y\}$.  

\paragraph{Types.} Each individual $i\in N$ is characterized by a \textbf{type}, which we denote by $t_i \in \mathcal{T} \equiv \mathbf{R}$, and the \textbf{type distribution} for group $g\in \mathcal{G}$ is described by a cumulative distribution function (CDF):
\[
F_g : \mathcal{T} \to [0,1].
\]
We assume throughout that $F_g$ is continuously differentiable and possesses full support on $\mathbf{R}$ (so that $F_g(t) \in (0,1)$ for all $g$ and $t$.

\paragraph{Behaviors and Signals.}  We consider binary choice settings: each individual $i \in N$ will choose a \textbf{behavior}, $\beta_i \in \mathcal{B} \equiv \{0,1\}$. Conditional on any individual $i$'s choice of $\beta_i$ and $i$'s group membership, $g_i$, a binary \textbf{signal}, $s_i\in \{0,1\}$ will be observed by the algorithm (which we describe below) according to a Bernoulli \textbf{sampling distribution}:
\[
\Pr[s_i=\beta_i\mid \beta_i] = \phi_{g_i} \in (\onehalf,1].
\]

\paragraph{Decisions and Algorithms.}  Each individual $i \in N$ will be assigned a binary \textbf{decision}, $d_i\in \mathcal{D} \equiv \{0,1\}$, based on his or signal, $s_i$ and possibly his or her group membership, $g_i$.  The algorithm is represented as a conditional distribution as follows:
\[
\Pr[d_i=s_i\mid s_i, g_i] = \delta_{s_i}(g_i) \in [0,1].
\]

\paragraph{Individual Preferences.}  Each individual $i \in N$ has a \textbf{payoff function}, $u: \mathcal{B} \times \mathcal{D} \times \mathcal{T} \rightarrow \mathbf{R}$.  To maximize comparability with the literature, we assume that individuals' payoff functions are independent of their group membership. Furthermore, we assume that $u$ is additively separable with respect to the behavior, $\beta_i$, and the decision, $d_i$, as follows:
\[
u(\beta_i, d_i, t_i) = r \cdot d_i - t_i \cdot \beta_i.
\]

\paragraph{Algorithmic Responsiveness.}  For each group $g \in \mathcal{G}$, any algorithm $\delta \in \mathcal{A}$'s \textbf{responsiveness for group $g$} is defined as follows:
\[
\rho_g(\delta,\phi_g) \equiv  \left(\delta_1(g)+\delta_0(g)-1\right)(2\phi_g-1).
\]
An algorithm's responsiveness is central to the individual incentives it induces, and it partitions $\mathcal{A}$ into three mutually exclusive and exhaustive categories as follows.
\begin{enumerate}
    \item \textbf{Positively Responsive Algorithms.}  When $\rho_g(\delta,\phi_g)>0$, the algorithm is more likely to award $d_i=1$ to agent $i$ when he or she chooses $\beta_i=1$ than when $\beta_i=0$
    \item \textbf{Negatively Responsive Algorithms.}  When $\rho_g(\delta,\phi_g)<0$, the algorithm is less likely to award $d_i=1$ to agent $i$ when he or she chooses $\beta_i=1$ than when $\beta_i=0$.
    \item \textbf{Null Algorithms.} When $\rho_g(\delta,\phi_g)=0$, whether agent $i$ receives a decision of $d_i=1$ is \textit{independent} of his or her choice of $\beta_i$.
\end{enumerate}
The special case of a null algorithm is equivalent (from a behavioral standpoint, at least) to having \textit{no algorithm at all}.

\subsection{Equilibrium Rewards \& Accuracy}
For any $(r,\delta,F,\phi)$, the equilibrium proportions of individuals in any group $g \in \mathcal{G}$ being assigned correct and incorrect decisions based on their behavior are defined by the following:
\begin{eqnarray*}
    W_1^*(g\mid r,\delta,F,\phi) & = & F_g(r\cdot \rho_g(\delta, \phi_g))(\phi_g \delta_1(g) + (1-\phi_g) (1-\delta_0(g))),\\
    W_0^*(g\mid r,\delta,F,\phi) & = & (1-F_g(r\cdot \rho_g(\delta, \phi_g))) ((1-\phi_g) (1-\delta_1(g)) + \phi_g \delta_0(g)),\\
    L_1^*(g\mid r,\delta,F,\phi) & = & (1-F_g(r\cdot \rho_g(\delta, \phi_g))) ((1-\phi_g) \delta_1(g) + \phi_g (1-\delta_0(g))), \text{ and}\\
    L_0^*(g\mid r,\delta,F,\phi) & = &  F_g(r\cdot \rho_g(\delta, \phi_g))(\phi_g (1-\delta_1(g)) + (1-\phi_g) \delta_0(g)).
\end{eqnarray*}
With the conditional probabilities defined above, the \textit{ex ante} equilibrium probabilities that an individual from group $g$ will be \textit{rewarded} (\textit{i.e.}, $d_i=1$, denoted by $R^*$) or \textit{punished} (\textit{i.e.}, $d_i=0$, denoted by $P^*$) are defined by the following:
\begin{eqnarray*}
    R^*(g\mid r,\delta,F,\phi) & = & W_1^*(g\mid r,\delta,F,\phi) + L_0^*(g\mid r,\delta,F,\phi), \text{ and}\\
    P^*(g\mid r,\delta,F,\phi) & = & L_1^*(g\mid r,\delta,F,\phi) + W_0^*(g\mid r,\delta,F,\phi).
\end{eqnarray*}

\section{Arrovian Statistical Discrimination \label{Sec:Arrow}}

In this appendix, we consider whether ERB or PP can eliminate \textit{clearly unfair} algorithms in the presence of algorithmic endogeneity.  As one might expect given that we wrote this article, the answer is, ``not necessarily.''  The argument along these lines is, we believe, slightly more subtle. And, to be clear, it's not our argument in substantive terms --- instead, it rests upon long acknowledged points with respect to the possibility that some forms of discrimination may be \textit{self-enforcing} (\cite{Arrow73}).  To make the point as simply as possible, we present one example of the inability of ERB and PP to identify a clearly discriminatory algorithm as being unfair.

\begin{example}[Arrovian Discrimination \& Fairness] 
\label{Ex:ArrovianDiscrimination}
Suppose that $r=1$, the type distribution is $F_g$, with mean $\mu_g$, for each $g \in \{X,Y\}$, and that the groups have identical levels of measurement noise: $\phi_X=\phi_Y = \phi \in (\onehalf,1]$.  Then, consider the following discriminatory algorithm (which violates AC):
\[
\tilde{\delta}_{s_i}(g_i) = \begin{cases}
    0 & \text{ if } s_i=1 \text{ and } g_i = X,\\
    1 & \text{ if } s_i=0 \text{ and } g_i = X,\\
    1 & \text{ if } g_i = Y,
\end{cases}
\]
The equilibrium prevalence for each group $g \in \{X,Y\}$ under algorithm $\tilde{\delta}$ is given by
\[
\pi_g(\tilde{\delta}) = \begin{cases}
    F_X(0) & \text{ if } g=X,\\
    F_Y(2\phi-1) & \text{ if } g=Y.
\end{cases}
\]
Accordingly, the algorithm $\tilde{\delta}$ produces the following error rates and predictive values in equilibrium:\\
\begin{table}[hbtp]
    \centering
    \begin{tabular}{|c|c|c|}
    \hline
            &\textbf{Group} $X$ &\textbf{Group} $Y$                 \\
    \hline
\textbf{False Positive Rate}& 0                 & $1-\phi$                  \\
\textbf{False Negative Rate}& $F_X(0)$  & $1-\phi$                  \\
    \hline
\textbf{Positive Predictive Value}& Undefined & $
\frac{F_Y\left(2\phi-1\right)\phi}{F_Y\left(2\phi-1\right)\phi+(1-F_Y\left(2\phi-1\right))\left(1-\phi\right)}
$\\
\textbf{Negative Predictive Value}&$1-F_X(0)$ &  $
\frac{\left(1-F_Y\left(2\phi-1\right)\right)\phi}{\left(1-F_Y\left(2\phi-1\right)\right)\phi+F_Y\left(2\phi-1\right)\left(1-\phi\right)}
$   \\
    \hline
    \end{tabular}
    \caption{Error Rates \& Predictive Values With Discrimination}
    \label{Tab:ArrovianExample}
\end{table}\\
Notice that $\tilde{\delta}$ generally does not satisfy ERB or PP and, indeed, $\tilde{\delta}$ \textit{might be more accurate for members of $X$ than it is for members of group $Y$}.

However, members of group $Y$ are clearly better off under $\tilde{\delta}$ than are members of group $X$.  This can be seen from the following calculations of type=-conditional expected payoffs under $\tilde{\delta}$:
\begin{equation}
\label{Eq:ArrovianWelfare}
    \begin{aligned}
    EU(t_i,g_i=X\mid \tilde{\delta},\phi,z) & = \max[0,-t_i],\\
    EU(t_i,g_i=Y\mid \tilde{\delta},\phi,z) & = \max[1-\phi,\phi-t_i].
    \end{aligned}
\end{equation}
\end{example}

\begin{example}[Self-Enforcing Discrimination Can Be Accurate]
    Continuing with the framework in Example \ref{Ex:ArrovianDiscrimination}, the discriminatory algorithm $\tilde{\delta}$ can be perfectly accurate, but only with respect to members of group $X$ --- \textit{the group $\tilde{\delta}$ discriminates against}!

    To see this, suppose that all members of group $X$ intrinsically prefer $\beta_i=0$ to $\beta_1$.  This implies that  $F_X(0)=0$.  Then, referring to Table \ref{Tab:ArrovianExample}, this implies that $\tilde{\delta}$ has zero errors (\textit{i.e.}, the false positive and false negative rates for members of $X$ are each equal to zero, and the negative predictive value for the same group is equal to 1).  

The same can not be true about $\tilde{\delta}$ for members of group $Y$, however.  In particular, because there is measurement noise ($\phi_Y<1$) and $\tilde{\delta}$ is non-trivially conditioning on $s_i$ for members of group $Y$, $\tilde{\delta}$ will necessarily make both positive and negative errors.  Furthermore, unless $F_Y(0)=1$ --- which would imply that all members of group $Y$ intrinsically prefer $\beta_i=1$ to $\beta_i=0$ --- no algorithm that assigns positive probability to awarding $d_i=1$ to members of group $Y$ can achieve zero error rates for members of group $Y$.
\end{example}

\begin{remark}
    \textit{If $F_X(0)=0$ and $F_Y(0)-1$, then there exists a} \textbf{perfect classifier} \textit{for each group $g$.  When this is the case, the results of \cite{KleinbergMullainathanRaghavan16} and \cite{Chouldechova17} imply that there is an algorithm that can simultaneously satisfy ERB and PP and, furthermore, it ignores each individual's signal, conditioning instead only on the individuals' group, $g\in \{X,Y\}$:}
    \[
    \delta(g,s) = \begin{cases}
        0 & \text{ if } g=X,\\
        1 & \text{ if } g=Y.
    \end{cases}
    \]
\end{remark}

\newpage

\section{General Classification Problems \label{Sec:GeneralClassification}}

The primitives are 
\begin{itemize}
    \item A set of \textbf{individuals}, $N$,
    \item A nonempty set of \textbf{types}, $\mathcal{T}$,
    \item A finite set of $G\geq 1$ \textbf{groups}, $\mathcal{G}$,
    \item A finite set of at least 2 \textbf{behaviors}, $\mathcal{B}$,
    \item A finite set of at least 2 \textbf{signals}, $\mathcal{S}$, 
    \item A finite set of at least 2 \textbf{decisions}, $\mathcal{D}$, 
    \item A set of $G$ \textbf{signal distributions}:
    \[
    \phi = \{\phi_g\}_{g \in \mathcal{G}},
    \]
    satisfying $\phi_g(s \mid \beta)\in [0,1]$ and $\sum_{s \in \mathcal{S}} \phi_g(s \mid \beta) = 1$ for all $g \in \mathcal{G}$, $s \in \mathcal{S}$, and $\beta \in \mathcal{B}$, and 
    \item A payoff function:
    \[
    u: \mathcal{D} \times \mathcal{B} \times \mathcal{T} \times \mathcal{G} \to \mathbf{R}.
    \]
\end{itemize}
For simplicity of presentation, we also impose the following property:
\begin{equation}
\label{Eq:Symmetry}
    |\mathcal{B}| = |\mathcal{S}| = |\mathcal{D}| = k \geq 2,
\end{equation}
and, without loss of generality, relabel these sets such that 
\[
\mathcal{B} = \mathcal{S} = \mathcal{D} = \{1,\ldots,k\}.
\]
Any problem satisfying \eqref{Eq:Symmetry} is a ($k$-)\textbf{symmetric classification problem} (or $k$-SCP).

\paragraph{Algorithms and Decisions.} A \textbf{classification algorithm} for any SCP is any function 
\[
\delta: \mathcal{S} \times \mathcal{G} \to \Delta(\mathcal{D}).
\]
where for any individual $j$ with signal $s_j \in \mathcal{S}$ and group membership $g_j$,
$\delta(s_j,g_j)$ denotes the conditional distribution of $d_j$. We denote the set of all algorithms for any $k$-SCP by $\mathcal{A}_k$.

\paragraph{Error Rate Balance.}  For any behavior $\beta \in \mathcal{B}$, we refer to any decision $d \in \mathcal{D}$ with $d \neq \beta$ as an \textbf{error}.  Thus, any classification algorithm for a $k$-SCP has $k!$ potential \textit{types of errors}.  On the other hand, if $d=\beta$, we say that the decision is \textbf{correct}.

There are at least two notions of an algorithm ``treating individuals equally well'' in a statistical sense.  The first of these notions requires that the probability that the algorithm gives an individual $i$ the correct decision be independent of $i$'s group membership.  We refer to this as \textbf{accuracy balance}.
\begin{definition}
\label{Def:AccuracyBalance}
    For any $k$-SCP, an algorithm $\delta \in \mathcal{A}$ satisfies \textbf{accuracy balance} (AB) if
    \[
    \Pr[d_i=\beta_i\mid \delta, g] = \Pr[d_i=\beta_i\mid \delta, g'] \;\;\; \text{ for all } \beta_i \in \mathcal{B} \text{ and } g,g' \in \mathcal{G}.
    \]
\end{definition}
A separate but related notion requires that conditional on any given behavior, $\beta_i$, the probability that the algorithm gives an individual $i$ \textit{any given decision} $d\in \mathcal{D}$ be independent of $i$'s group membership.  This is simply the general form of \textbf{error rate balance}.
\begin{definition}
    For any $k$-SCP, an algorithm $\delta \in \mathcal{A}$ satisfies \textbf{error rate balance} (ERB) if
    \[
    \Pr[d\mid \beta_i, \delta, g] = \Pr[d\mid \beta_i, \delta, g'] \;\;\; \text{ for all } d \in \mathcal{D}, \beta_i \in \mathcal{B}, \text{ and } g,g' \in \mathcal{G}.
    \]
\end{definition}
ERB is a stronger requirement than AB in $k$-SCPs unless $k=2$.
\begin{remark}
\textit{While ERB clearly implies AB in any $k$-SCP, the converse is true if and only if $k=2$.}    
\end{remark}

\paragraph{Individual Preferences.}  We normalize all individuals' payoffs to equal zero conditional on the decision being correct.  Individual $i$'s payoff function can be written as
\[
u : \mathcal{B} \times \mathcal{D} \times \mathcal{T} \times \mathcal{G} \to \mathbf{R},
\]
satisfying
\[
u(x,x,t,g) = 0 \text{ for any } x \in \mathcal{B} = \mathcal{D}, t \in \mathcal{T}, \text{ and } g \in \mathcal{G}.
\]
Individual $i$'s conditional expected payoff from behavior $\beta \in \mathcal{B}$, given $t_i$, $g_i$, $\delta$, and $\phi$, is
\[
EU(\beta\mid t_i, g_i, \delta, \phi) = \sum_{s\in \mathcal{S}} \phi_{g_i}(s\mid \beta) \cdot \bigg[ \sum_{d\in \mathcal{D}} 
 \delta(d \mid s,g_i) \cdot u(\beta,d,t_i,g_i) \bigg],
\]
and individual $i$'s \textbf{best response}, given $t_i$, $g_i$, $\delta$, and $\phi$, is any function $b^*$ satisfying
\[
b^*(t_i,g_i\mid \delta,\phi) \in B^*(t_i,g_i\mid \delta,\phi) \equiv \argmax_{b \in \mathcal{B}} EU(\beta\mid t_i, g_i, \delta, \phi).
\]
so that, finally, $i$'s \textbf{conditional expected payoff from} $\delta$, given $t_i$, $g_i$, and $\phi$, is
\[
EU^*(\delta\mid t_i, g_i, \phi) = EU(b^*(t_i,g_i\mid \delta,\phi)\mid t_i, g_i, \delta, \phi).
\]

\begin{definition}
    An algorithm $\delta$ satisfies \textbf{equal opportunity} (EO) if 
    \[
    B^*(t,g\mid \delta,\phi) \cap B^*(t,g'\mid \delta,\phi) \neq \emptyset \;\;\; \text{ for all } t\in \mathcal{T}, g, g' \in \mathcal{G}, \text{ and all } i,j \in N.
    \]
\end{definition}

\begin{definition}
    An algorithm $\delta$ is \textbf{envy free} if 
    \[
    EU^*(\delta\mid t_i, g_i, \phi) \geq EU^*(\delta\mid t_i, g', \phi) \;\;\; \text{ for all } i\in N \text{ and } g' \in \mathcal{G}.
    \]
\end{definition}

\begin{definition}
A payoff function $u$ is \textbf{group independent} if 
\[
u(\beta,d,t,g) = u(\beta,d,t,g') \;\;\; \text{ for all } g,g' \in \mathcal{G} \text{ and all } \beta\in \mathcal{B}, d \in \mathcal{D}, and t \in \mathcal{T}.
\]    
\end{definition}

If $\delta$ satisfies EF and $u$ satisfies group independence, then $\delta$ satisfies EO.  The converse does not hold, because EO is based only on the \textit{difference} between expected payoffs from $\delta_i=0$ and $\delta_i=1$.  In generally interesting situations, it is possible to hold the difference between these expected payoffs constant while having an algorithm that adjusts the expected probability of getting the reward based on group membership such that every member of any given group would make the same choices, but many individuals would have an incentive to claim that they are a member of a different group (aside from members of ``the most favored group'' under $\delta$).

\begin{theorem}
\label{Th:EnvyFreeERB}
    If $u$ satisfies group independence, then 
    \begin{enumerate}
        \item Envy freeness of $\delta$ implies EO of $\delta$ in any $k$-SCP,
        \item ERB of $\delta$ implies envy freeness of $\delta$ in any $k$-SCP, and
        \item Envy freeness of $\delta$ implies ERB of $\delta$ in a $k$-SCP if and only if $k=2$.
    \end{enumerate}
\end{theorem}
\begin{proof}
    Fix any payoff functions, $u$, satisfying group independence, and any $k$-SCP algorithm $\delta \in \mathcal{A}$. 
    \begin{itemize}
        \item (Envy Free $\Rightarrow$ EO)  Anonymity of $u$ implies that 
    \[
    EU(\beta\mid t_i, g, \delta, \phi) = EU(\beta\mid t, g, \delta, \phi)
    \]
    for all $i, j\in N, t \in \mathcal{T}, \beta \in \mathcal{B}, g \in \mathcal{G}, \text{ and } \delta \in \mathcal{A}$ and, if $\delta$ is envy free, then 
    \[
    EU^*(\delta\mid t_i, g_i, \phi) \geq EU^*(\delta\mid t_i, g', \phi) \;\;\; \text{ for all } i\in N \text{ and } g' \in \mathcal{G}.
    \]
    Satisfaction of group independence by $u$ and envy freeness of $\delta$ imply that 
    \begin{equation}
    \label{Eq:EnvyFreeEquality}
    EU^*(\delta\mid t_i, g_i, \phi) = EU^*(\delta\mid t_i, g', \phi). 
    \end{equation}
    Accordingly, $B^*(t,g\mid \delta,\phi) = B^*(t,g'\mid \delta,\phi)$. Accordingly, satisfaction of envy freeness by $\delta$ implies that $\delta$ also satisfies EO.
    \item (ERB $\Rightarrow$ Envy Free)  Fix any $k$-SCP algorithm $\delta$. 
 Individual $i$'s expected payoff from choosing $\beta_i$, given $t_i$, $g_i$, and $\delta$, is
        \begin{equation}
            \label{Eq:ERBImpliesEnvyFreeStep1}
            EU(\beta_i\mid \delta, t_i, g_i) = \sum_{s\in \mathcal{S}} \sum_{d \in \mathcal{D}} \bigg[ \phi_{g_i}(s\mid \beta_i) \cdot \delta(d \mid s,g_i) \cdot u(d,\beta_i,t_i,g_i) \bigg],
        \end{equation}
        so that the supposition that $u$ is group independent allows us to reduce Equation \ref{Eq:ERBImpliesEnvyFreeStep1} to
        \begin{equation}
            \label{Eq:ERBImpliesEnvyFreeStep2}
            EU(\beta_i\mid \delta, t_i, g_i) = \sum_{s\in \mathcal{S}} \sum_{d \in \mathcal{D}} \bigg[ \phi_{g}(s\mid \beta_i) \cdot \delta(d \mid s,g) \cdot u(d,\beta_i,t_i) \bigg],
        \end{equation}
        Satisfaction of ERB by $\delta$ in any $k$-SCP implies that 
        \[
        \sum_{s\in \mathcal{S}} \sum_{d \in \mathcal{D}} \bigg[ \phi_{g}(s\mid \beta) \cdot \delta(d \mid s,g) \bigg]
        = \sum_{s\in \mathcal{S}} \sum_{d \in \mathcal{D}} \bigg[\phi_{g'}(s\mid \beta) \cdot \delta(d \mid s,g')\bigg],
        \]
        for each $\beta \in \mathcal{B}$ and all $g,g' \in \mathcal{G}$.  More precisely, satisfaction of ERB by $\delta$ in a $k$-SCP implies that 
        \begin{equation}
        \label{Eq:ERBImpliesEnvyFreeStep3}
        \phi_{g}(d \mid \beta) \cdot \delta(d \mid s,g)
        = \phi_{g'}(d \mid \beta) \cdot \delta(d \mid s,g')
        \end{equation}
        for each $\beta \in \mathcal{B}$, all $d \in \mathcal{D}$, and all $g,g' \in \mathcal{G}$.  Thus, let 
    \begin{equation}
        \label{Eq:PsiDefinition}
    \psi_g(d,\beta,\delta) \equiv \phi_{g}(d \mid \beta) \cdot \delta(d \mid s,g).
    \end{equation}
    Then satisfaction of ERB by $\delta$ implies $\psi_g(d,\beta,\delta) = \psi_{g'}(d,\beta,\delta) \equiv \bar{\psi}(d,\beta,\delta)$ for all $g,g' \in \mathcal{G}$, and Equation \eqref{Eq:ERBImpliesEnvyFreeStep2} can be rewritten as follows:
        \begin{equation}
            \label{Eq:ERBImpliesEnvyFreeStep4}
            EU(\beta_i\mid \delta, t_i, g_i) = \sum_{d \in \mathcal{D}} \bigg[ \bar{\psi}(d,\beta,\delta) \cdot u(d,\beta_i,t_i) \bigg],
        \end{equation}
        The right hand side of Equation \ref{Eq:ERBImpliesEnvyFreeStep4} is independent of $g_i$, implying that whenever preferences are group independent, satisfaction of ERB by $\delta$ implies that $\delta$ satisfies both envy freeness and equal opportunity.
        
    \item (Envy Free $\Rightarrow$ ERB)  Equation \eqref{Eq:EnvyFreeEquality} holds for any algorithm $\delta$ satisfying EF.  For any $k$-SCP, Equation \eqref{Eq:EnvyFreeEquality} can be reexpressed as
    \[
    \sum_{d \in D} \Pr[d\mid b^*(t\mid \delta), \delta, g] u(d,b^*(t\mid \delta),t,g) = \sum_{d \in D} \Pr[d\mid b^*(t\mid \delta), \delta, g'] u(d,b^*(t\mid \delta),t,g'),
    \]
    for all $t \in \mathcal{T}$, and all $g,g' \in \mathcal{G}$.  Now consider any $t \in \mathcal{T}$ for which 
    \begin{equation}
    \label{Eq:BreakIndifference}
        u(1,b^*(t\mid \delta),t) > u(0,b^*(t\mid \delta),t).
    \end{equation}
    (This will be sufficient for our purposes when $k=2$ and will also provide insight into why EF implies ERB only in that case.) Simplifying notation, let $u^*(d) \equiv u(d,b^*(t\mid \delta),t)$.  We now proceed to show that this implication holds if and only if $k=2$.
    
    \begin{itemize}
         \item ($k=2$.)
    If $k=2$, then Equation \eqref{Eq:EnvyFreeEquality} can be further reduced to the equality of the following for all groups, $g \in \mathcal{G}$:
    \[
    EU_g(b^*(t)) \equiv (\phi_g \delta_1(g)  + (1-\phi_g)(1-\delta_0(g))) u^*(1) + ((1-\phi_g) \delta_0(g)  + \phi_g(1-\delta_1(g))) u^*(0). 
    \]
    We can simplify the presentation by normalizing the expected payoffs as follows:
    \begin{eqnarray*}
        & & (\phi_g \delta_1(g)  + (1-\phi_g)(1-\delta_0(g))) u^*(1) + ((1-\phi_g) \delta_0(g)  + \phi_g(1-\delta_1(g))) u^*(0) \\
        & \propto & (\phi_g \delta_1(g)  + (1-\phi_g)(1-\delta_0(g))) (u^*(1)-u^*(0)),
    \end{eqnarray*}
    so EF implies
    \[
    \phi_g \delta_1(g)  + (1-\phi_g)(1-\delta_0(g))
    \]
    is equal across all groups $g \in \mathcal{G}$, which is equivalent to equality of true positive rates for all groups.  An analogous derivation establishes that EF implies equality of true negative rates for all groups.  Accordingly, if $u$ is group independent, then satisfaction of EF by $\delta$ implies satisfaction of ERB by $\delta$.  
        \item ($k>2$.)  If $k>2$, we can provide a counterexample as follows.  Specifically, suppose $k=3$ and that $u$ is group independent.  Consider any $t \in \mathcal{T}$ with 
        \[
        u(0,0,t) > u(1,0,t) > u(2,0,t),
        \]
        and without loss of generality, suppose that $u(1,0,t)\equiv 0$.  (Notice that we have omitted $g$ because $u$ is group independent.)  Using the notation defined above in \eqref{Eq:PsiDefinition}, consider the family of algorithms defined as follows for any $\sigma\in [0,1]$ and $\alpha_g \in (0,1)$:
        \[
            \psi_g(d,0,\delta) = \begin{cases}
                \sigma \alpha_g & \text{ if } d=0,\\
                1-\alpha_g & \text{ if } d=1,\\
                (1-\sigma) \alpha_g & \text{ if } d=2.
            \end{cases}
        \]
        Suppose that $\alpha_g=0$ and $\alpha_{g'}=1$.  Then the following algorithms are envy free conditional on choosing $\beta_i=0$:
        \[
        \bigg\{\sigma: u(0,0,t) \sigma + u(2,0,t) (1-\sigma) \bigg\}.
        \]
        This is a continuum of algorithms, and relative to that continuum, they generically violate ERB.  Finally, note that this example applies to any $k>2$.
        \end{itemize}
    \end{itemize}
    We have shown that envy freeness implies equal opportunity and that ERB implies envy freeness in any $k$-SCP.  In addition, satisfaction of ERB implies envy freeness in a $k$-SCP if and only if $k=2$, as was to be shown.
\end{proof}
Theorem \ref{Th:EnvyFreeERB} is informative in several ways.  The most notable connection it illuminates regards the relationship between ERB and envy freeness: ERB is sufficient for envy freeness in \textit{all} $k$-SCPs.  On the other hand, the converse is true if and only in binary classification problems ($k=2$).  

More subtly, the theorem illuminates a distinction between behavior and welfare.  Envy freeness implies EO --- when the algorithm does not induce any individual to want to change their group membership prior to classification, then the algorithm must induce individuals in different groups to use identical strategies when choosing $\beta_i$.  However, the converse does not hold.  
This implies that an algorithm that induces EBR between two groups is not necessarily envy free.  

The power of this can be seen by a clearly discriminatory algorithm.  Suppose that members of group $X$'s costs are distributed according to a $\mathrm{Uniform}(0,1)$ distribution, and members of group $Y$'s costs are distributed according to a $\mathrm{Uniform}(0,2)$ distribution, and all individuals (regardless of whether they are in group $X$ or group $Y$) value the decision $d_i=1$ at $w\geq 2$, and $\phi_X=\phi_Y=1$.  Then the algorithm
\begin{eqnarray*}
\delta_0(X) = \delta_0(Y) & = & 0,\\
\delta_0(X) & = & 0.5,\\
\delta_0(Y) & = & 1,
\end{eqnarray*}
satisfies EBR (both members of group $X$ and $Y$ choose $\beta_i=1$ with certainty, regardless of $t_i$), but this algorithm clearly violates envy freeness: members of group $Y$ have unambiguously greater expected payoffs under $\delta$ than do members of group $X$.  Yet, the payoffs of the two groups satisfy anonymity and group independence.  

\section{Asymmetric Signal Structures}

In this section, we relax our assumption that the conditional distribution of $s$ given $\beta_i$ is independent of $\beta_i$.  To reduce notation, we consider only one group.  Define $k$ conditional distributions --- one for each of the behavioral choices in $\mathcal{B}$ --- as follows:
\[
\Pr[s_i=s\mid \beta_i=\beta] = \phi_\beta(s).
\]
To keep the presentation simple, this can be written in the binary classification case ($\mathcal{B}=\{0,1\}$) simply as follows:
\begin{eqnarray*}
\Pr[s_i=\beta_i\mid \beta_i=0] & = &  \phi_0 \in [0,1], \text{ and}\\
\Pr[s_i=\beta_i\mid \beta_i=1] & = &  \phi_1 \in [0,1].
\end{eqnarray*}
Whenever $\phi_0\neq \phi_1$, the signal structure generates different rates of Type-I and Type-II errors.  

\begin{table}[hbtp]
\centering
\begin{tabular}{|c|c|c|}
    \cline{2-3}
    \multicolumn{1}{c|}{}   &  $d_i=1$          & $d_i=0$                   \\
    \hline
         $\beta_i=1$        & $W_i$               & $0$    \\
    \hline
         $\beta_i=0$        & $\omega_i+t_i$    & $t_i$                     \\
    \hline
    \end{tabular}
    \caption{Preferences Over Outcomes \label{Tab:ErrorPreferencesRaw2}}
\end{table}
Fix any binary classification algorithm, $\delta \in [0,1]^2$, and suppose that $u(d,\beta,t)$ is defined as in Table \ref{Tab:ErrorPreferencesRaw2}. Individual $i$'s expected payoffs from each $\beta_i\in \{0,1\}$ are as follows:
\begin{eqnarray*}
    EU(\beta_i=0\mid t_i) & = & \phi_0 (\delta_0 t_i + (1-\delta_0) (\omega_i+t_i)) + (1-\phi_0) (\delta_1 (\omega_i + t_i) + (1-\delta_1) t_i),\\
    & = & t_i + (\phi_0 (1-\delta_0) + (1-\phi_0) \delta_1) \omega_i,\\
    EU(\beta_i=1\mid t_i) & = & \phi_1 (\delta_1 W_i + (1-\delta_1) 0) + (1-\phi_1) (\delta_0 (0) + (1-\delta_0) W_i),\\
    & = & (\phi_1 \delta_1 + (1-\phi_1) (1-\delta_0)) W_i
\end{eqnarray*}
The IC condition for $i$ to choose $\beta_i=1$ is then
\begin{eqnarray*}
    EU(\beta_i=0\mid t_i) & \leq & EU(\beta_i=1\mid t_i),\\
    t_i + (\phi_0 (1-\delta_0) + (1-\phi_0) \delta_1) \omega_i& \leq & (\phi_1 \delta_1 + (1-\phi_1) (1-\delta_0)) W_i,\\
    \Rightarrow t_i & \leq & (\phi_1 \delta_1 + (1-\phi_1) (1-\delta_0)) W_i - (\phi_0 (1-\delta_0) + (1-\phi_0) \delta_1) \omega_i.
\end{eqnarray*}
Suppose that $W_i=\omega_i\equiv r$: 
\begin{eqnarray*}
    t_i & \leq & (\phi_1 \delta_1 + (1-\phi_1) (1-\delta_0)) r - (\phi_0 (1-\delta_0) + (1-\phi_0) \delta_1) r,\\
    \Rightarrow t_i & \leq & r \cdot (\delta_0+\delta_1-1)(\phi_1+\phi_0-1) ,
\end{eqnarray*}
which reduces to the following when $\phi_0=\phi_1=\phi$:
\begin{eqnarray*}
    t_i & \leq & r \cdot  (\delta_0+\delta_1-1)(2\phi-1).
\end{eqnarray*}

\newpage

\bibliography{john-Madrid-Fairness}

\end{document}